\documentclass[12pt]{article}

\usepackage{amsmath}
\usepackage{amssymb}
\usepackage{amsthm}
\usepackage{mathrsfs}

\usepackage{slashed}
\usepackage{color}
\usepackage{esint}

\usepackage{graphicx}
\usepackage[all,cmtip]{xy}
\usepackage{hyperref}
\usepackage{subcaption}

\usepackage{authblk}
\makeatletter
\let\@fnsymbol\@alph
\makeatother

\newtheorem{thm}{Theorem}
\newtheorem{lemma}{Lemma}
\newtheorem{prop}{Proposition}

\theoremstyle{definition}
\newtheorem{defn}{Definition}
\newtheorem{examp}{Example}

\newtheorem{rmk}{Remark}

\providecommand{\keywords}[1]{{{Keywords:}} #1}

\hypersetup{
    pdfborder = {0 0 0},%
    colorlinks,%
    citecolor=blue,%
    filecolor=black,%
    linkcolor=red,
    urlcolor=green
}

\def\H{\mathscr H}
\def\Z{\mathbb{Z}}
\begin{document}

\title{Topological insulators and K-theory}

\author[1]{Ralph M. Kaufmann  \thanks{e-mail: rkaufman@math.purdue.edu} }
\author[1]{Dan Li  \thanks{e-mail: li1863@math.purdue.edu } }
\author[2]{Birgit  Wehefritz--Kaufmann   \thanks{e-mail: ebkaufma@math.purdue.edu} }

\affil[1]{Department of Mathematics, Purdue University, 150 N University St\\
         West Lafayette, IN 47907, USA }
\affil[2]{Department of Mathematics and Department of Physics and Astronomy, Purdue University\\
		150 N University St,          West Lafayette, IN 47907, USA }

\date{}
\maketitle

\begin{abstract}
We analyze the topological $\mathbb{Z}_2$ invariant,
which characterizes time reversal invariant topological insulators, in the framework of  index theory and K-theory.
The topological $\mathbb{Z}_2$ invariant counts the parity of generalized Majorana zero modes, which can be interpreted as an analytical index.
As we show, it fits perfectly into a mod 2 index theorem, and the topological index provides an efficient way to compute the topological $\mathbb{Z}_2$ invariant.
Finally, we give a new version of the bulk-boundary correspondence which yields
an alternative explanation of the index theorem and the topological $\mathbb{Z}_2$ invariant.
Here the boundary is not the geometric boundary of a probe, but an effective boundary in the momentum space.

\end{abstract}


\section{Introduction}\label{Intro}

 As new materials observed in nature,  topological insulators behave like
insulators in the bulk but have conducting edge  states on the boundary.
A time reversal invariant topological insulator, or simply topological insulator \cite{HK10}, of free (or weak-interacting) fermions, is a quantum
system that has both time reversal $\mathbb{Z}_2$ symmetry  and charge
conservation  (a $U(1)$ symmetry).  Accordingly, in the framework of symmetry protected topological (SPT) phases \cite{CGLW13},
 topological insulators are also referred to as $U (1)$ and time reversal symmetry protected topological order.
According to the Altland-Zirnbauer-Cartan classification of general topological insulators with different $\mathbb{Z}_2$ symmetries \cite{AZ97, SRFL09},
topological insulators with time reversal symmetry are in the symplectic class,
which is also referred to as  type AII topological insulators.
Topological insulators are gapped quantum systems, in which the Fermi energy level is assumed to be in the band gap, 
and further the Fermi level is defined as the zero energy level.

Topological insulators
can be characterized by a $\mathbb{Z}_2$-valued topological  invariant due to the time reversal symmetry. 
The main task of this work is to understand the topological $\mathbb{Z}_2$ invariant in the framework of index theory and K-theory. 
We introduce the notion of generalized   Majorana zero modes and show that one interpretation  of the $\mathbb{Z}_2$ invariant is given by
the parity of these generalized Majorana zero modes. The generalized Majorana zero modes are similar to those found in  Bogoliubov--de
Gennes (BdG) topological superconductors with particle-hole symmetry.  
We will drop the word ``generalized'' in generalized Majorana zero modes from now on.
Our definition of Majorana zero modes generalizes that of Majorana fermions in physics \cite{W09},
but they are related by generalizing the self-adjointness condition (or the real condition).
Physically, Majorana zero modes \cite{DFN15} are quasi-particle excitations bound to a defect at zero energy. 
Two fermionic  Majorana zero modes tend to  couple together and behave as an effective composite boson.
For example, Dirac cones are Majorana zero modes in three-dimensional (3d) time reversal invariant fermionic systems,
and the presence of an unpaired Dirac cone is the characteristic  of a 3d non-trivial topological insulator.
We have to point out that Majorana zero modes should not be confused with Majorana fermions. Majorana fermions have been understood
by the spin representations and spin geometry \cite{LM90}.
 Majorana zero modes  have a new  geometry beyond spin geometry compared to Majorana fermions. 
 The study of topological insulators is basically to understand the topology of time reversal symmetry and the geometry of Majorana zero modes. 

The topological $\mathbb{Z}_2$ invariant is a parity anomaly in quantum field theory, as a global anomaly, in general it is difficult to compute.
The idea is to translate the parity anomaly into a gauge anomaly, and consider a relevant  index problem.
We will explain that the topological $\mathbb{Z}_2$ invariant can be simply understood via a mod 2 index
theorem resolving the gauge anomaly.
 More precisely, the mod 2 index theorem can be spelled out as  (1) that the analytical $\mathbb{Z}_2$ index of the effective Hamiltonian counts the parity of Majorana zero modes
 and (2) that
the odd topological  index (or its variation) gives a local formula to compute the analytical index.
In a modern language, the mod 2  index theorem can be reformulated as the index pairing between   $KR$-homology and $KR$-theory, 
which will be useful for further generalizations in noncommutative geometry.

Physicists originally proposed two different ways to model this $\mathbb{Z}_2$-valued   invariant.
The Kane--Mele invariant  was first defined in the quantum spin Hall effect in graphene \cite{KM05},
and subsequently generalized to 3d  topological
insulators  \cite{FKM07}. It is  defined as the product of the signs of Pfaffians over the fixed points of the time reversal symmetry.
In contrast to the discrete Pfaffian formalism of the Kane--Mele invariant, the Chern--Simons invariant has an integral form
as the odd topological index of a specific gauge transformation induced by the time reversal symmetry \cite{QHZ08}.
In string theory, the Chern--Simons invariant is also called the Wess--Zumino--Witten (WZW) topological term.
The WZW term is an action functional in gauge theory, and the Kane--Mele invariant is derived from an effective quantum field theory.
In fact, the Kane--Mele invariant and the Chern--Simons invariant are equivalent, so we view them as different aspects of the same $\mathbb{Z}_2$ invariant,
which is the so-called topological $\mathbb{Z}_2$  invariant.
For a summary of these results, see e.g. \cite{KLW15}.

In the various proposed ways, the $\Z_2$ invariant takes values in different abstract groups.
For instance in the 3d case, in several identifications $\Z_2$ is just a factor of a product,
with other factors corresponding to weak topological insulators,
or other times a factor $\Z$ appears that is disregarded, see e.g.\ \cite[Theroem 11.14]{FM13} where the relevant group is $\Z\oplus \Z_2^{\times 4}$ and
\cite[Equation (27)]{K09} where the relevant group is $\Z\oplus \Z_2^{\times 3}$.
In our treatment, we naturally obtain a unique receptacle for the $\Z_2$ invariant.
The analytical and topological index both point out the $\mathbb{Z}_2$ invariant really lives in $KO^{-2}(pt)$.
That is, the abstract group, in which the $\Z_2$ invariant lies, is isomorphic just to $\Z_2$.
Index theory has the advantage over K-theory in that
 local geometric pictures can be seen from an index theorem by looking at, for example, the spectral flow. 
 The fact that the topological $\mathbb{Z}_2$ invariant belongs to $KO^{-2}(pt)$  cannot be seen from  K-theory only.

In order to understand the mod 2 index theorem, we study the version of $K$-theory  relevant for the systems of
topological insulators under study. It is the Quaternionic K-theory, i.e., $KQ$-theory, since the time reversal symmetry introduces a real structure.  
This furnishes the right framework to study the index theory of Majorana zero modes.  The Quaternionic $K$-theory is related to 
the Real K-theory, i.e.,  $KR$-theory, by a degree shifting isomorphism, but distinct from the real $K$-theory, i.e., $KO$-theory.
We briefly review the relations also in \S\ref{quatsec}, more details can be found in \cite{Karoubi}. The reason that $KQ$ appears naturally is that if one considers
the topological band theory of a fermionic system, one is actually working on a Hilbert bundle over the momentum space,
which is called the Brillouin zone in condensed matter physics.
Taking time reversal symmetry into account, the Hilbert bundle becomes  a Quaternionic vector bundle equipped 
with anti-involutions which do not fix the base space, but are rather compatible with an involution on the base space.
Hence $KQ$-theory can be used to classify all possible  band structures of type AII topological insulators.

In spin geometry and noncommutative geometry, KR-cycles are used to model spinors with real structures.
A KR-cycle gives a canonical representative of the fundamental class in KR-homology, which is closely related to the KO-orientation.   
As a (Bogoliubov) quasi-particle, a Kramers pair consists of an electron and its mirror partner under time reversal symmetry. 
The localization of a Kramers pair gives a localized Majorana zero mode (around a fixed point). 
An electronic chiral state in a Majorana zero mode can be modeled by a KR-cycle,
so  a Majorana zero mode is described by a coupled product of two KR-cycles.

As an instance of the holographic principle,  the bulk-boundary correspondence plays an important role in topological insulators. 
If the bulk and boundary theory are modeled by K-theory, then the bulk-boundary correspondence gives rise to a homomorphism (not necessary to be an isomorphism)
between the bulk $K$-theory and the boundary $K$-theory.
In other words, the bulk-boundary correspondence falls into the category of bivariant K-theory, i.e., KK-theory. 
The equivalence between the mod 2 topological index and the Kane--Mele invariant gives hints on the geometry of the bulk-boundary correspondence. 
The bulk is given by the momentum space, and our novel observation is that the effective boundary can be identified as the set of fixed points,
since the Kane--Mele invariant is defined over the fixed points. 
Based on the Baum--Connes isomorphism for the torus, We  explain the bulk-boundary correspondence by a concrete example.
In fact, the bulk-boundary correspondence in a topological insulator realizes  the same  mod 2 index theorem.

In our approach, we focus on the topological band theory, the geometry of Majorana zero modes, and the index theorem connecting these components.
Different independent, concurred or subsequent  approaches using different techniques can be found in the literature.
These include approaches via $C^*$-algebras including a definition for topological insulators in that setting \cite{K15}, Roe algebra \cite{K16},
twisted crossed product \cite{T16}, extensions and $KK$-theory \cite{BCR16,BKR16} and lattice models \cite{KK16}. 
All the approaches, including ours, yield topological invariants that are robust, 
which is stable under small perturbations such as disorder.
An alternative bulk-boundary correspondence in momentum space is studied in \cite{MT16}  
by using a Fourier--Mukai transform in the guise of T-duality applied to the real space bulk-boundary correspondence.

Calculations for these types of invariants are usually done using long exact sequences. Their existence is tied to the topological properties of the base space. 
To this end we include a careful analysis of the conditions on the base space and the $\Z_2$ action on it. 
With different restrictions, we show that different long exact sequences become available and we link them to different definitions of the $\Z_2$ invariant that have appeared in the literature.

In summary, in this work, we give a complete explanation of
the topological $\mathbb{Z}_2$ invariant in the framework of index theory and $K$-theory.
We provide the following new results on topological insulators:
\begin{itemize}

 \item  By localization, we interpret the topological $\mathbb{Z}_2$ invariant as the mod 2 analytical index of the effective Hamiltonian of a topological insulator.

 \item  By comparing   Majorana zero modes with   Majorana fermions in spin geometry, we model localized Majorana zero modes by
 a coupled product of two $KR$-cycles. 

 \item We establish the mod 2 index theorem of topological insulators, that is, the odd topological index (or its variation) computes the analytical index.  
Moreover, the mod 2 index can be obtained by the index pairing between $KR$-homology and $KR$-theory.

 \item We derive the equivalence between the topological $\mathbb{Z}_2$ index (or Chern--Simons invariant) and the Kane--Mele invariant, which
 is  the  genuine bulk-boundary correspondence.

 \item We propose to identify the effective  boundary as the fixed points of the time reversal symmetry,
 and explain the bulk-boundary correspondence by a motivating example. 
\end{itemize}

This article is organized in the following way. To study  topological band theory,
we  review  basic facts  about KQ-theory in \S \hyperref[sec:KQ]{2}. This section also includes the conditions on the base space and the respective sequences in K-theory.
\S \hyperref[sec:Aind]{3} focuses on the analytical theory of  Majorana zero modes, and  the topological $\mathbb{Z}_2$ invariant is interpreted as
a mod 2 analytical index. In \S \hyperref[sec:Tind]{4}, the mod 2 topological index is discussed in KR-theory, and the mod 2 index theorem is established. Finally, the bulk-boundary
correspondence in KK-theory is  discussed in \S \hyperref[sec:Bbcorr]{5}.

\section{KQ-theory}\label{sec:KQ}

In this section, we will first introduce the  time reversal  symmetry and study the topological band theory of a  topological insulator by a Hilbert
bundle, which is a Quaternionic vector bundle over the momentum space. As a result, the Quaternionic K-theory, i.e., $KQ$-theory, can be used  to classify all possible band structures of a topological insulator.
Finally, the real Baum--Connes isomorphism for the free discrete group $\mathbb{Z}^d$ will be briefly reviewed as a preparation for the bulk-boundary correspondence.

\subsection{Time reversal symmetry}

Let $X$ be a compact space,  which is viewed as the momentum space of  a topological insulator.
\begin{examp}
A lattice in $\mathbb{R}^d$ is a free abelian group isomorphic to $\mathbb{Z}^d$ and its Pontryagin dual is the torus $\mathbb{T}^d$,
 the simplest and most important example of a momentum space is  $X = \mathbb{T}^d$.
\end{examp}
\begin{examp}
The limit of a lattice model is the continuous model defined on $\mathbb{R}^d$ and its Pontryagin dual is itself,
in this case the momentum space is the one point compactification of $\mathbb{R}^d$, i.e.,  $X = (\mathbb{R}^d)^+ = \mathbb{S}^d$.
\end{examp}

\begin{defn}
An involutive space $(X, \tau)$ is a compact space $X$ equipped with an involution, i.e., a homeomorphism $\tau: X \rightarrow X$
such that $\tau^2 = id_X$.
\end{defn}
Since the involution $\tau$ is usually taken as the complex conjugation, $(X, \tau)$ is also called a Real space, 
which was first introduced by Atiyah in the  Real K-theory, i.e., $KR$-theory  \cite{A66}.
As a convention,  define $ \mathbb{R}^{p,q} : = \mathbb{R}^{p} \oplus i \mathbb{R}^{q}$,  the involution $\tau$ on $\mathbb{R}^{p,q}$ is defined by the complex conjugation, 
 i.e., $\tau : (x, y) \rightarrow (x, -y)$, or equivalently $\tau|_{\mathbb{R}^{p}} = 1$ and  $\tau|_{i \mathbb{R}^{q}} = -1$. 

Time reversal symmetry is a fundamental symmetry of physical laws, which is the  $\mathbb{Z}_2$ (as a notation $\mathbb{Z}_2 \cong \mathbb{Z}/2 \mathbb{Z}$) symmetry  defined by the
 map $T:   t \mapsto -t$ reversing the direction of time.
 The action of time reversal symmetry
 on a momentum space is basically to change the sign of its local coordinates  because $T$ changes the sign of the imaginary unit $T:   i \mapsto -i$.
\begin{defn}
 Time reversal symmetry defines the time reversal transformation on the momentum space $X$,
 \begin{equation*}
   \tau: X \rightarrow X; \quad   {x} \mapsto \tau({x}) = -x
 \end{equation*}
so that $(X, \tau)$ is an involutive space.
\end{defn}
   The time reversal transformation may vary according to the choice of coordinate system on $X$.

\begin{examp}
   The unit sphere $\mathbb{S}^d$ in Cartesian coordinates is
   $$
   \mathbb{S}^d = \{ (x_0, x_1, \cdots x_d) \in \mathbb{R}^{d+1} ~|~ x_0^2 + x_1^2 + \cdots x_d^2 = 1\}
   $$
   the time reversal transformation on $\mathbb{S}^d$ is defined by
   $$
   \tau: \mathbb{S}^d \rightarrow \mathbb{S}^d; \quad  (x_0, x_1, \cdots x_d) \mapsto  (x_0, -x_1, \cdots -x_d)
   $$
   so the unit sphere under this involution is also denoted by $ \mathbb{S}^{1, d} \subset \mathbb{R}^{1, d}$.
\end{examp}

\begin{examp}
   If the torus $\mathbb{T}^d$ is parametrized by the angles,
   $$
   \mathbb{T}^d = \{ (e^{i\theta_1}, \cdots , e^{i\theta_d}) ~|~ \theta_i \in [-\pi, \pi ]\, mod\,\, 2\pi, i = 1, \cdots d \}
   $$
   then the time reversal transformation on $\mathbb{T}^d$ is defined by
   $$
   \tau:  \mathbb{T}^d \rightarrow \mathbb{T}^d; \quad (e^{i\theta_1}, \cdots , e^{i\theta_d}) \mapsto (e^{-i\theta_1}, \cdots , e^{-i\theta_d})
   $$
\end{examp}

The fixed points of an involution $\tau$ is the set of points
\begin{equation*}
   X^\tau := \{ {x} \in X ~|~ \tau( {x} ) =  {x} \}
\end{equation*}
which is always assumed to be a finite set.
\begin{examp}
 The unit sphere $\mathbb{S}^{1,d}$ has $2$ fixed points under the time reversal transformation, $(\mathbb{S}^{1,d})^\tau = \{ (\pm1, 0, \cdots, 0) \}$.
\end{examp}

\begin{examp}
   The torus $\mathbb{T}^d$ has $2^d$ fixed points under the time reversal transformation,  $(\mathbb{T}^d)^\tau = \{ (\pm1, \pm1, \cdots, \pm1) \}$ when $\theta_i = 0, \pi$.
\end{examp}

 The involutive space $(X, \tau)$ has the structure of a $\mathbb{Z}_2$-CW complex,
 that is, there exists a $\mathbb{Z}_2$-equivariant cellular decomposition of $X$, see below.
 In the other way around,
 starting with the fixed points $X^\tau$, $X$ can be built up by gluing cells that carry a free $\mathbb{Z}_2$ action, i.e., $\mathbb{Z}_2$-cells.
 Such equivariant cellular decomposition is very useful in the computation of  K-theory.
 This construction is closely related to the stable homotopy splitting of $X$ into spheres respecting the time reversal symmetry \cite{FM13}. 
 
  We will now decompose the underlying space $X$ according to the action by the time reversal transformation $\tau$.
In the general setting, this decomposition can be quite wild, but in concrete situations it is usually well behaved. For instance, as we explain in \S\ref{seqsec},
in case this decomposition is suitably ``nice'' we can apply different long exact sequences to compute the invariants.

We will assume that the Real space $(X, \tau)$ is {\em tame}. For a connected $X$ this means that
there  are closed connected fundamental domains $V_{\pm}$, such that $\tau(V_\pm)=V_\mp$, $X=V_+\cup V_-$
and $B=V_+\cap V_-$ is the closed boundary of both $V_+$ and $V_-$. That is  $V_\pm=V^o_\pm\amalg B$ as sets,
with  $\tau(V^o_{\pm})=V^o_{\mp}$ open and $B$ is closed of
codimension greater or equal to 1. Such $B$ separates, namely,
 $V^o_+$ and $V^o_-$ occupy different components of $X\setminus X^\tau$.
 Here and in the following $\amalg$ means the disjoint union. For a general $X$, being tame means that each connected component of $X$ is tame.
This is for instance the case for a Riemannian manifold $X$ where
$V_\pm$ are given by Dirichlet fundamental domains, a.k.a., Voronoi cells. We call a tame space {\em regular}, if
we can find a decomposition $X=X_+\amalg X_-\amalg X^{\tau}$
where $\tau (X_{\pm})=X_{\mp}$ such that $\bar X_\pm=V_\pm$,  $X_\pm$ and $X\setminus X^\tau$  are locally compact.
This is the case for all the examples that we will consider including the Examples given above,  see Figure \ref{torusfig}. 

\begin{figure} 

\includegraphics[width=4cm]{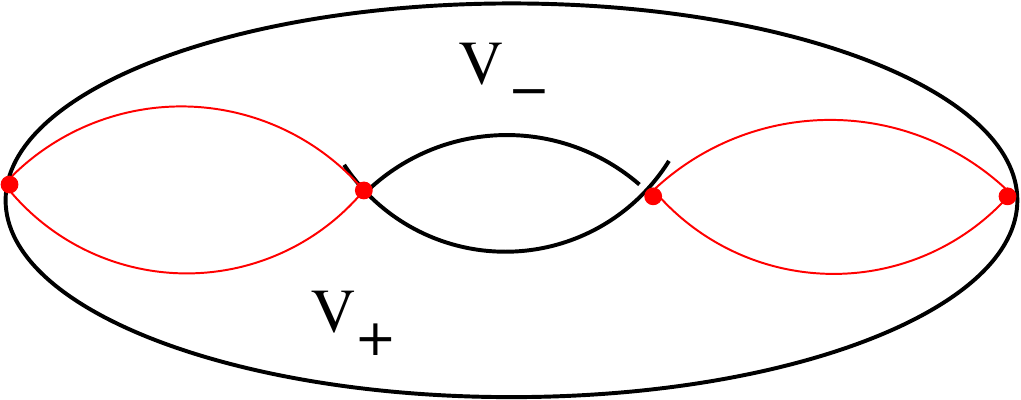}\hfill
\includegraphics[width=2cm]{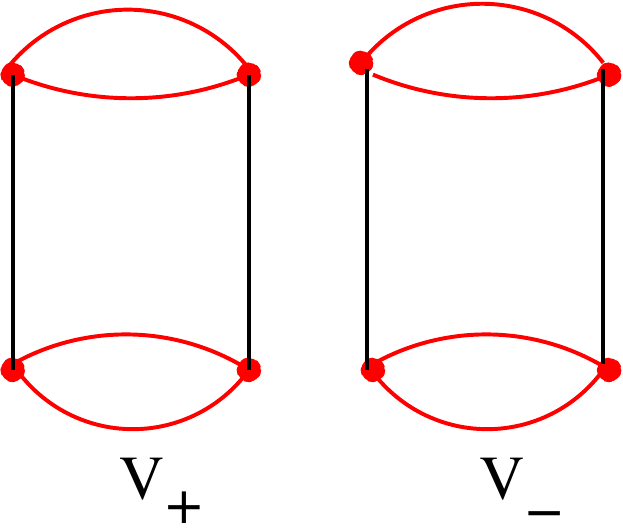} \hfill
\includegraphics[width=2cm]{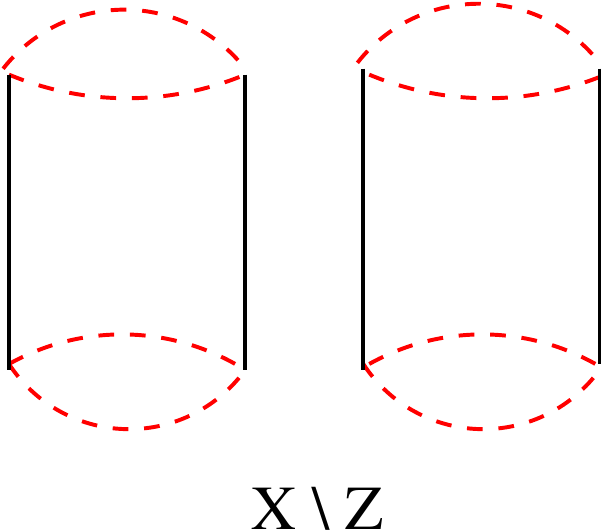} \hfill
\includegraphics[width=2cm]{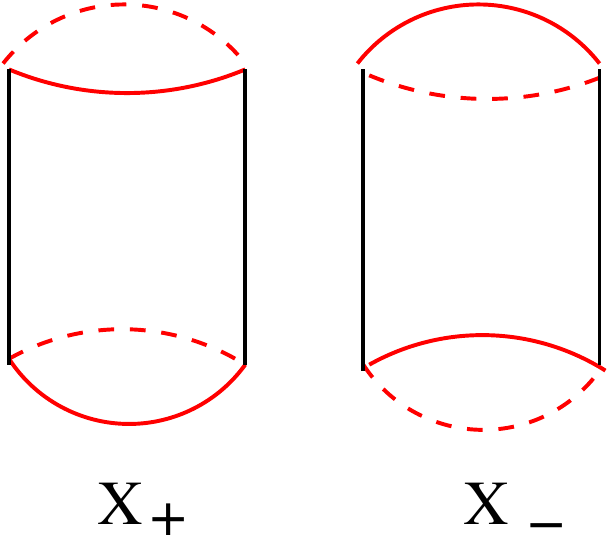}

\caption{\label{torusfig}From the left to right: (a) the torus with two fundamental domains $V_+,V_-$ with
$Z=V_+\cap V_-=\mathbb{S}^{1,1}\amalg \mathbb{S}^{1,1}$ in red,
(b) the disjoint union of $V_+$ and $V_-$, (c) the complement of $Z$, and (d) the subsets $X_\pm$. Dotted lines indicate open boundaries.}
\end{figure}

Another class, which is the most important for applications is when $(X,\tau)$ is a finite $\Z_2$-equivariant regular CW complex. This means that for all cells $C$ of dimension $k$,  $\tau(C)$ is also a cell of dimension $k$.
  In this case, we can decompose $X=V_+\cup V_-$ as above where now $V_\pm$ are sub CW complexes.
Moreover, $V_\pm=\bar O_\pm$ where $O_\pm=\amalg_{C\in V_\pm} C^o$ is the union of all the interiors of the top dimensional cells in $V_\pm$,
that is, those that are not at the boundary of any other cell. We denote the closed cells by $C$ and write $C^o$ for their open interior.

There is also another decomposition, which we will use $X=X_+\amalg X_-\amalg X^{\tau}$.
To do this, we assign $+$, $-$ or $fix$ to all cells,  inductively, by choosing fundamental domains as above starting
with the dimension zero cells and using induction on the $k$-skeleton. We choose $+$ and $-$ for the cells interchanged by $\tau$ and $fix$ for all the cells fixed by $\tau$.
We will call them $+$, $-$ or fixed cells.  The induction ensures that no $+$ cell lies at the boundary of only $-$ cells.
Notice that  fixed points do not lie in the interior of any $+$ or $-$ cell, moreover,   $X^{\tau}=\amalg_{C: \text{fixed cell}}C^o $.
Set $X_+=\amalg_{C: + \text{ cell}}C^o$,
$X_-=\amalg_{C: - \text{ cell}}C^o $.
We call such a CW complex  {\em  weak $\mathbb{Z}_2$-space}  if there is a choice of $\pm$ such that each skeleton $X^k$ is regular,
as defined above, with  respect to the decomposition above.  This is for instance the case, if $X$ is a compact manifold and $X^{\tau}$ is discrete, which encompasses all the examples from the literature.
In particular this means that if $Z=V_+\cap V_-$ then $Z$ is again a weak $\mathbb{Z}_2$-space and
one can use induction.

\begin{examp} All the $\mathbb{T}^d$ and $\mathbb{S}^{1,d}$ are of this type. For $d=0$ the space consists of two points, marked by $fix$. For $d=1$ one adds two intervals joining the points,  marked by $+$ and $-$.
For $\mathbb{S}^{1,d}$, we realize it as $\mathbb{R}^{0,d}\cup \{\infty\}$.
Mark $\infty$  by $fix$ and then mark $\mathbb{R}^{0,d}$ by decomposing it w.r.t.\ the iterated upper and lower half spaces, marking the upper half space by $+$ and the lower by $-$.
This agrees with the decomposition as $\mathbb{R}^{0,d}=\mathbb{R}^{0,1}\times \dots\times \mathbb{R}^{0,1}$.
Similarly, we can define the decomposition of $\mathbb{T}^{d}=\mathbb{S}^{1,1}\times \dots \times \mathbb{S}^{1,1}$, see Figure \ref{ABfig}.
\end{examp}

\begin{figure}
\includegraphics[width=2cm]{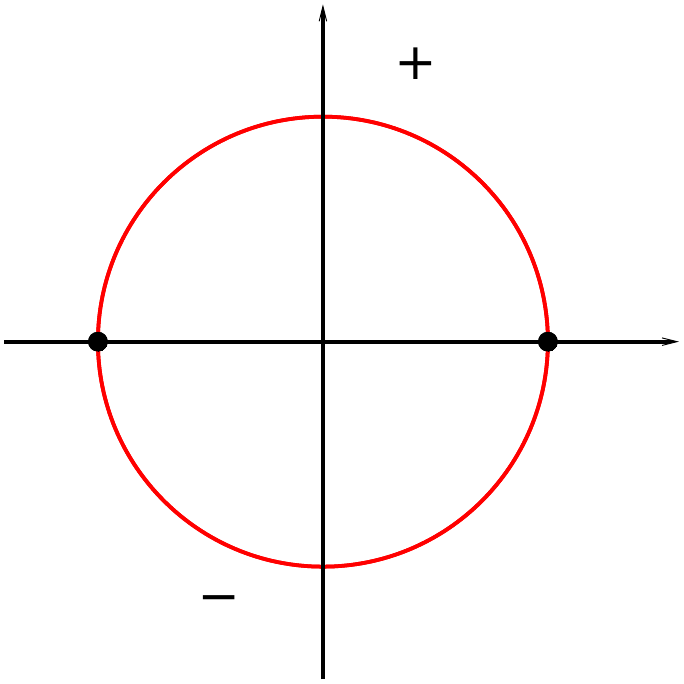}\hfill
\includegraphics[width=2cm]{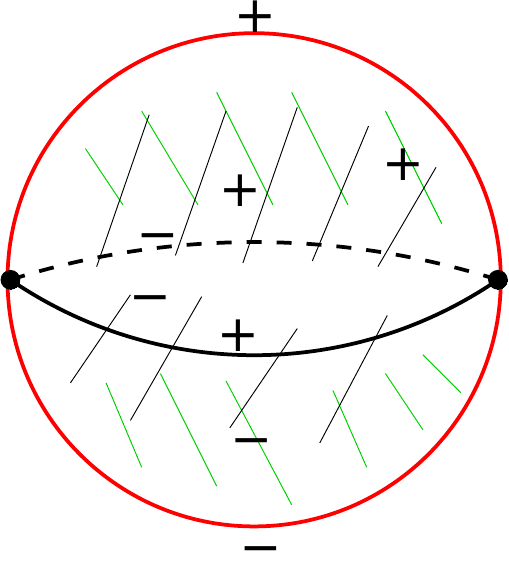}\hfill
\includegraphics[width=2cm]{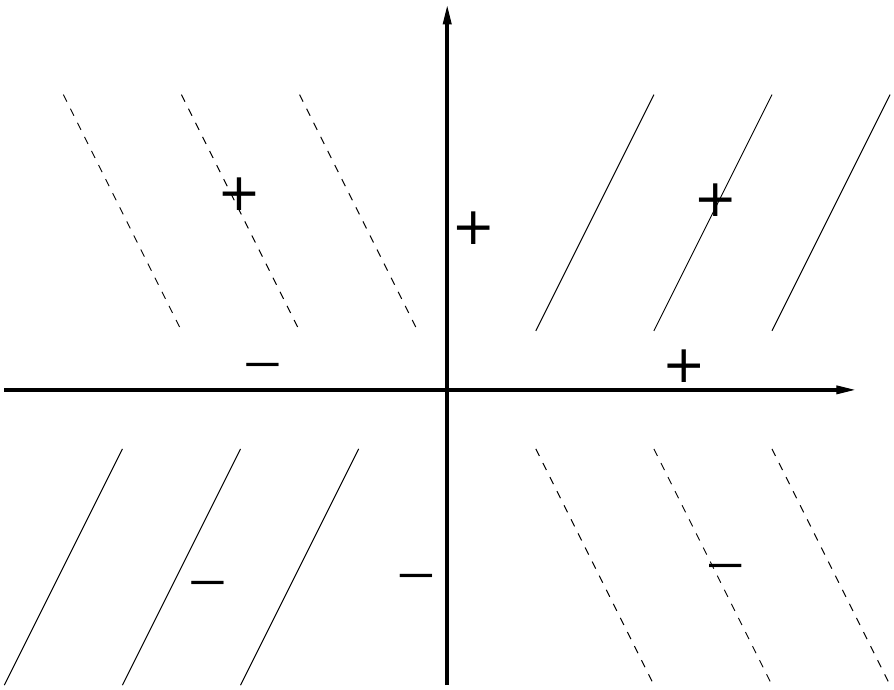}\hfill
\includegraphics[width=2cm]{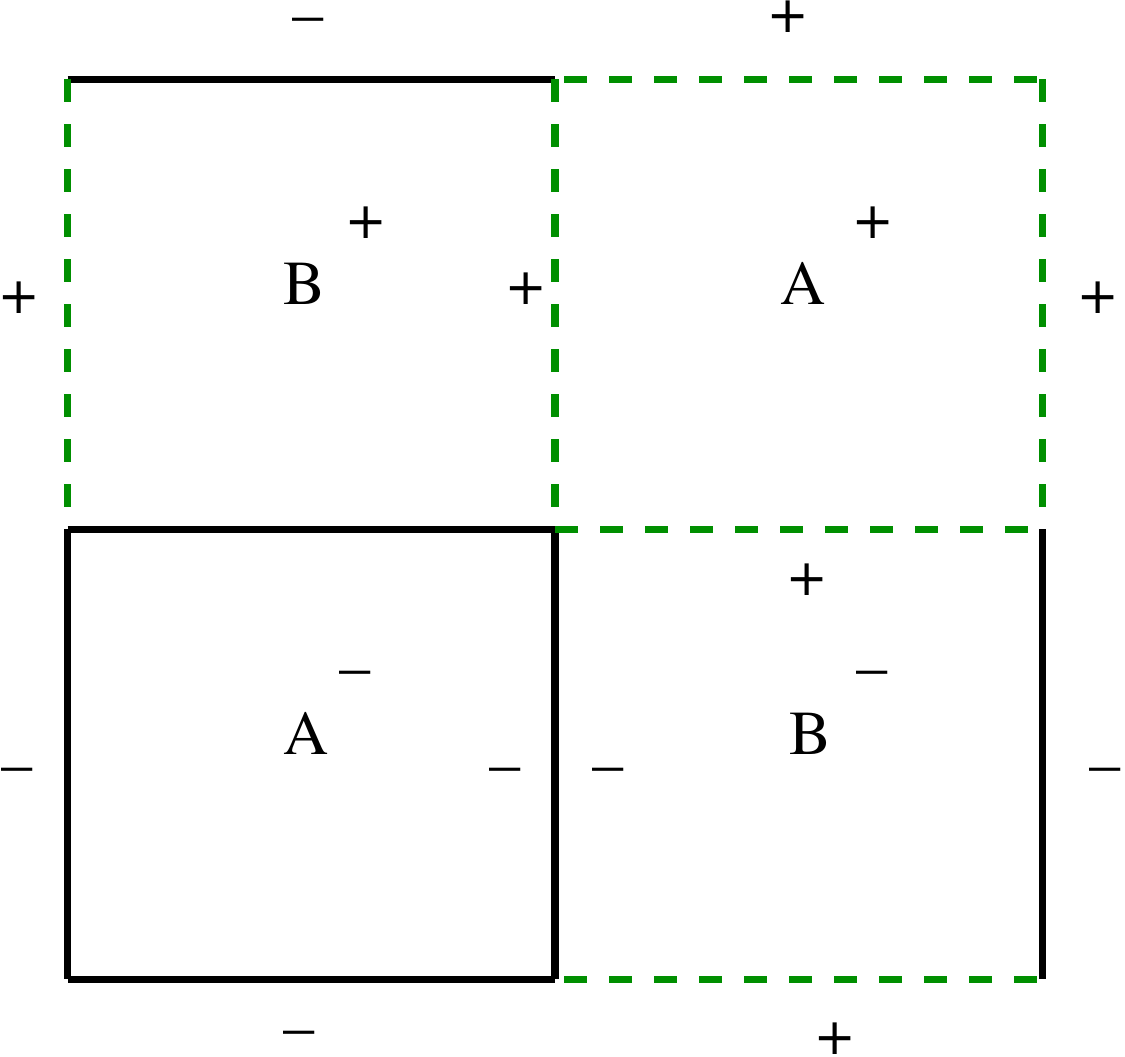}
\caption{\label{ABfig} The $\mathbb{Z}_2$-CW decompositions of (a) $\mathbb{S}^{1,1}$,(b)
$\mathbb{S}^{1,2}$, with 2 0-cells, 4 1-cells and 4 2-cells,
corresponding to the one-point compactification of $\mathbb{R}^{0,2}$ decomposed into quadrants and half lines depicted in (c), and (d) the decomposition of $\mathbb{T}^{2}$.}
\end{figure}

\begin{defn}
   A complex vector bundle is said to be a Hilbert bundle if  it is equipped with a complete continuously varying Hermitian inner product so that  each fiber is a Hilbert space. The completeness is automatic in the finite dimensional case.
\end{defn}

A Hilbert bundle (with a flat connection) was introduced to model continuous fields of Hilbert spaces in geometric quantization, for example see \cite{ADW91}.
Given an effective Hamiltonian of  a topological insulator,
let us consider the corresponding Hilbert bundle $\pi: \mathcal{H} \rightarrow X$,
which describes the band structure of the topological insulator over the momentum space $X$.
The space of physical states (in the momentum representation) is modeled by the Hilbert space of (local) sections $\H=\Gamma(X,\mathcal{H})$ with fiber-wise inner product.

\begin{defn}
 A Quaternionic vector bundle $(E, \chi)$ over an involutive space $(X, \tau)$ is a complex vector bundle $E$ over $X$ equipped with an
 anti-linear anti-involution $\chi$ (that is compatible with $\tau$). The anti-involution $\chi$ is an anti-linear bundle isomorphism  $\chi: E \rightarrow E$  such that $\chi^2 = -id_E$.
This entails that $\chi$ induces an anti--linear isomorphism between the fibers  $E_x$ over $x$ and    $E_{\tau(x)}$ over $\tau(x)$.
\end{defn}

In the above definition, if the anti-involution $\chi$ (s.t. $\chi^2 = -1$) is replaced by an involution $\iota$
(s.t. $\iota^2 = 1$), then $(E, \iota)$ is a Real space and
the pair $(E, \iota)$ defines a Real vector bundle over $(X, \tau)$.

For a Quaternionic vector bundle $\pi: (E, \chi) \rightarrow  (X, \tau)$, the compatibility condition between $\chi$ and $\tau$ is
obviously $\pi \circ \chi = \tau \circ \pi$.
For any section $s: (X, \tau) \rightarrow (E, \chi)$, the above compatibility condition implies that
for any $x \in X$, $s \circ \tau(x)$ and $\chi \circ s(x)$ are in the same fiber.
In order to compare two sections,  we define an action $\varsigma$ on the space of continuous sections $\Gamma(X, E)$,
$$
\varsigma: \Gamma(X, E) \rightarrow \Gamma(X, E); \quad s \mapsto - \chi \circ s \circ \tau
$$
$\varsigma$ is itself an anti-involution such that $\varsigma^2 = -1$.
The negative sign in the definition of $\varsigma$ comes from the anti-involution $\chi$ in the following way.
Comparing the above two sections, one would have the difference $s\circ\tau-\chi\circ s$.
Moving this back to the original fiber by applying $\chi$, the difference becomes $\chi\circ s\circ\tau+s=s-\varsigma(s)$.

At a fixed point  $x \in X^\tau$, the restriction $\chi_x: E_x \rightarrow E_x$ is an
anti-linear map such that $\chi_x^2 = -1$ so that $\chi_x$ gives  a quaternionic structure on $E_x$, i.e., a fiber preserving action of the quaternions. If the complex dimension of $E_x$ is even, say $\dim_\mathbb{C}E_x = 2n$,
then $E_x$ can be viewed as
a vector space defined over the quaternions $\mathbb{H}$ with quaternionic dimension $\dim_\mathbb{H}E_x = n$. From now on, we
 assume that the Quaternionic vector bundle $(E, \chi)$ has even rank, say $rank(E) = 2n$.
Let $i: X^\tau  \hookrightarrow X$ be the inclusion map, when restricting $E$ to the set of fixed points,
$i^*E \rightarrow X^\tau$ turns into a quaternionic vector bundle, i.e., each fiber is a quaternionic vector space.
The reader should not confuse quaternionic vector bundles with Quaternionic vector bundles.

Over the fixed points $X^\tau$, if an inner product is chosen on the quaternionic vector bundle
$i^*E \rightarrow X^\tau$, i.e., a Hilbert bundle with the quaternionic structure induced by $\chi$, then it naturally gives rise to a symplectic structure $\omega$
so that $(i^*E, \omega) \rightarrow X^\tau$ becomes  a symplectic vector bundle, i.e., each fiber is a symplectic vector space, see e.g.\ \cite{KLW15}[3.3.1].

 Time reversal symmetry can be represented by the time reversal operator
 $\Theta$ acting on the Hilbert space $\H$.
  In general, the  time reversal operator can be defined as
 a product  $\Theta := UC$, where $U$ is a unitary operator and $C$ is the complex conjugation \cite{Wigner}. Hence, $\Theta$ is  an anti-unitary operator, that is, for $\psi, \phi \in \Gamma(X, \mathcal{H})$,
\begin{equation*}
  \langle \Theta \psi, \Theta \phi \rangle = \overline{\langle  \psi,  \phi  \rangle} = \langle  \phi, \psi \rangle, \quad \quad \Theta (a \psi + b \phi) = \bar{a} \Theta\psi + \bar{b} \Theta \phi, \,\,\,\, a, b\in \mathbb{C}
\end{equation*} Since $\Theta$ is acting on fermionic states,
it has the important property $\Theta^2 = -1$ (or $\Theta^* = - \Theta$).

\begin{examp}
 Taking  time reversal symmetry into account, the Hilbert bundle over the momentum space $\pi: \mathcal{H} \rightarrow X$ becomes
   a Quaternionic Hilbert bundle $\pi: (\mathcal{H}, \Theta) \rightarrow (X, \tau)$, where $\tau$ is the time reversal transformation
   on the base space
   (s.t. $\tau^2 = id_X$) and $\Theta$ is  the time reversal operator (that is an anti-unitary operator satisfying $\Theta^2 = - id_\mathcal{H}$).
   The action of time reversal symmetry on the momentum space $X$  is given by $x\mapsto \tau(x)$,
 which can be lifted to an anti-linear anti-involution $\Theta$ on the Hilbert bundle $\mathcal{H} \rightarrow X$. 
\end{examp}


A time reversal invariant Hamiltonian  $H=H(x)$  must satisfy the condition
   \begin{equation*}
   \label{hameq}
  \Theta  H(x) \Theta^{*} = H(\tau(x)) \quad
    \text{for any $x \in X$}
    \end{equation*}
For  a physical electronic state $\phi$, 
one considers the other state $\Theta \phi$ under time reversal symmetry. If $\phi$ is an eigenstate of $H$ satisfying  $H(x)\phi(x)=E(x)\phi(x)$, 
then $\Theta \phi$ is an eigenstate of $\Theta H \Theta^*$ with the same energy $E$, 
$$
 [\Theta H(x) \Theta^*] \Theta \phi (x)  = E(\tau(x))\Theta \phi (x) 
$$
This means that the states $\phi$ and $\Theta\phi$ fall into the same band with the band function $E: X \rightarrow \mathbb{R}$. 
So due to time reversal symmetry,  each band is doubly degenerate, which is called Kramers degeneracy, and the pair $(\phi,\Theta \phi)$ is called a Kramers pair.
The key property $\Theta^2 = -1$  implies that a Kramers pair is orthogonal, i.e., $\langle \phi, \Theta \phi\rangle=0$. 

The following facts are discussed in detail  in \cite{KLW15}.
With time reversal symmetry  the rank of the Hilbert bundle $\mathcal{H}$ is always assumed to be even, say $rank( \mathcal{H}) = 2N$.
We also assume that time reversal symmetry is the only symmetry acting on the Hilbert bundle. So the Hilbert bundle is non-degenerate
in the sense that there is no band crossings between different bands. For each band, the only degeneration is due to Kramers degeneracy,
that is, each band is doubly degenerate. All possible intersections between a Kramers pair only could happen at the fixed points.
As a consequence, there exists a decomposition
\begin{equation}\label{decompeq}
\mathcal{H} = \bigoplus_{i=1}^N \mathcal{H}_i
\end{equation}
where $\mathcal{H}_i \rightarrow X$ is a rank 2 non-trivial subbundle. We set $\H_i=\Gamma(X,\mathcal{H}_i)$,
then $\H=\bigoplus_{i=1}^N\H_i$ and $\Theta$ respects this decomposition.

If there was no time reversal symmetry, then each band  can be modeled by a complex line bundle $\mathcal{L}_i \rightarrow  X$. 
When the time reversal symmetry is switched on, each band  is modeled by a rank 2 Hilbert subbundle 
$\pi: \mathcal{H}_i \rightarrow X$. By the existence of a Kramers pair    $(\phi_i, \Theta \phi_i)$,  $\mathcal{H}_i$ can be decomposed into   
\begin{equation}
 \mathcal{H}_i = \mathcal{L}_i \oplus \Theta{\mathcal{L}}_i, \quad \text{with} \quad    \Theta{\mathcal{L}}_i \cong \tau^*(\overline{\mathcal{L}}_i) 
\end{equation}
where  $\tau^*(\overline{\mathcal{L}}_i)$ is the pullback of the conjugate bundle since $\Theta$ is an anti-linear bundle isomorphism. 
Inspired by the physical picture, we assume $\phi_i$ (or $\Theta \phi_i$) gives rise to a global section of $\mathcal{L}_i$ (or $\Theta \mathcal{L}_i$), so that $\mathcal{L}_i$  (or $\Theta \mathcal{L}_i$) is isomorphic to a trivializable line bundle.
However, the Hilbert subbundle $\mathcal{H}_i$ is not necessary to be a trivial bundle since there are intersections between $\phi_i$ and $\Theta \phi_i$. 
By the decomposition of $\mathcal{H}_i$, it is similar to a spinor bundle, since the time reversal operator $\Theta$ switches the chirality of a Kramers pair. 
In the literature, some authors assume the Hilbert bundle is a trivial complex vector bundle, for example see \cite{G17}.

\begin{lemma}
   Define the   transition function of the Hilbert subbundle $\mathcal{H}_i$  by
   \begin{equation} \label{TransFunc}
      w^i: X   \rightarrow U(2); \quad w^i(x) := \begin{pmatrix}
                                                  0 &       \langle \phi_i({x})\, , \, \varsigma \phi_i ({x})\rangle  \\
                                                      \langle \Theta \phi_i({x})\, , \,  \varsigma (\Theta \phi_i) ({x})\rangle &    0
                                                 \end{pmatrix}
   \end{equation}
   then the time reversal symmetry  is  represented by $w^i \circ C$ on  $\mathcal{H}_i$, where $C$ is the complex conjugation.
\end{lemma}
\begin{proof}
  The anti-involution $\varsigma$ acting on  $\Gamma(X, \mathcal{H}_i)$ is defined by
    \begin{equation*}
              \varsigma(\psi) = -\Theta \circ \psi \circ \tau, \quad \psi \in \Gamma(X, \mathcal{H}_i)
     \end{equation*}
     The space of sections $\Gamma(X, \mathcal{H}_i)$ is generated by  a pair of continuous global sections  $(\phi_i, \Theta \phi_{i})$ (i.e., a Kramers pair). 
     Since the Kramers pair $(\phi_i, \Theta \phi_{i})$ is orthogonal, i.e.,   
     $$
     \langle \phi_i({x})\, , \, \Theta \phi_i ({x})\rangle = 0 
     $$ the diagonal terms in $ w^i$ vanish. By the physical property of time reversal symmetry, $\Theta \phi_i(x)$ is the same as $\phi_i(\tau(x))$ up to a phase, so the off-diagonal terms
     in $ w^i$ keep track of the phases, for more details see \cite{KLW15}.
\end{proof}

In physics, the transition function $w^i$ is also called the (global) gauge transformation induced by the time reversal symmetry, so  
the subbundle $\mathcal{H}_i$ has the gauge group  $U(2)$. The transition function is the key to constructing a local formula to compute the
topological $\mathbb{Z}_2$ invariant.
Based on $w^i$, the parity anomaly of the topological $\mathbb{Z}_2$ invariant is translated into a gauge anomaly of the topological band theory. 
In the literature, $w^i$ is sometimes called a sewing matrix \cite{G17}.

Let us choose an open subset $O \subset X$,
then the local isomorphism $\Theta : \mathcal{H}_i|_{O } \rightarrow \mathcal{H}_i|_{\tau(O) } $ is represented by
\begin{equation*}
   \Theta:  O \times \mathbb{C}^2 \rightarrow  \tau(O)  \times \mathbb{C}^2, \quad (x, v) \mapsto (\tau({x}), w^i(x) \, \bar{v})
\end{equation*}
Apply $\Theta$ twice to get back to $O$,
\begin{equation*}
   (x, v) \mapsto (\tau({x}), w^i(x) \, \bar{v}) \mapsto (x, w^i(\tau({x}))\bar{w}^i(x) \, {v})
\end{equation*}
because of $\Theta^2 = -1$, we have
\begin{equation}
   w^i(\tau({x}))\bar{w}^i(x) = -I_2, \quad \text{i.e.,} \quad  [w^i]^T(\tau({x})) = -{w}^i(x), \quad \forall \,\, x \in X
\end{equation}
where $T$ stands for the  transpose of a matrix.
In particular,  $w^i$ is skew-symmetric at any fixed point,
\begin{equation}
   [w^i]^T(x) =- w^i ({x}), \quad \forall \,\, x \in X^\tau 
\end{equation}

\begin{examp}
  When $X = \mathbb{S}^3 = \{ (\alpha, \beta) \in \mathbb{C}^2, \, s.t. \, |\alpha|^2 + |\beta|^2 = 1 \}$,
   the time reversal transformation is defined by $\tau (\alpha, \beta) = (\bar{\alpha}, - \beta)$,
  and the fixed points are    $(\alpha, \beta) = (\pm 1, 0)$.
  In addition, the transition function $w$ is given by
  $$
  w: \mathbb{S}^3 \rightarrow SU(2); \quad ( \alpha, \beta) \mapsto \begin{pmatrix}
                                                                     \beta & \alpha \\
                                                                     -\bar{\alpha} & \bar{\beta}
                                                                    \end{pmatrix}
  $$
\end{examp}

The total transition function over the Hilbert bundle $\mathcal{H}$ is defined by the block diagonal matrix $w = diag\{w_1, \cdots, w_n \}$,
i.e., $w: X \rightarrow U(2n)$. Furthermore, a non-degenerate Quaternionic vector bundle is characterized
by the transition function $w$. However, the characteristic feature of a topological insulator is completely determined by the top band (indeed the edge states in the band gap),
so we tend to assume the Hilbert bundle has rank two.

\subsection{Quaternionic K-theory}\label{quatsec}

\begin{defn} For a compact Real space $(X, \tau)$, the Quaternionic K-group $KQ(X, \tau)$
is defined to be the Grothendieck group of finite rank Quaternionic vector bundles $(E, \chi)$ over  $(X, \tau)$.
\end{defn}

From the previous subsection, we know that the Hilbert bundle of a topological insulator is a finite rank Quaternionic vector bundle over
the momentum space, which describes the band structure in the presence of the time reversal symmetry.
Note that the Quaternionic K-group $KQ(X, \tau)$ classifies stable isomorphism classes of Quaternionic vector bundles over $X$.
On an  even rank trivial bundle  $X \times \mathbb{C}^{2k}$, there exists a natural quaternionic  structure acting on the fibers \cite{DG15}.
Hence using the Whitney sum to add a trivial bundle makes perfect sense,  and one can take the stable isomorphism classes of Quaternionic vector bundles.
In fact, adding a trivial bundle does not change the physics of a topological insulator.

When the involution $\tau$ is understood, the Quaternionic K-group $KQ(X, \tau)$ is denoted simply  by $KQ(X)$. 
$KQ$-theory can be extended onto locally compact spaces,  and the reduced  Quaternionic K-group  $\widetilde{KQ}(X)$ is defined as the kernel of the restriction map $ i^*: KQ(X) \rightarrow KQ(pt) = \mathbb{Z}$.

The  Real K-group $KR(X, \tau)$ is similarly defined as the Grothendieck group of Real vector bundles $(E, \iota)$ over $(X, \tau)$, which was first introduced by Atiyah \cite{A66}. 
Higher KR-groups are defined by
$$
KR^{-p, -q}(X) := KR(X \times \mathbb{R}^{p, q})
$$
There exists an isomorphism, called the $(1,1)$-periodicity of KR-theory,
$$
KR^{p+1, q+1}(X) \cong KR^{p,q }(X)
$$
so by convention a KR-group is denoted by $KR^{p-q}(X) =KR^{p,q }(X) $. The Bott periodicity of KR-theory is 8,
$$
KR^{n+ 8 }(X) \cong KR^{n}(X)
$$
There exists a canonical isomorphism 
$$
KQ^*(X) \cong KR^{*-4}(X)
$$
so it is convenient to compute $KQ$-theory by  $KR$-theory.
As a result, the Bott periodicity of KQ-theory is easily derived from that of KR-theory.

When the involution $\tau$ in KR-theory is trivial, it becomes the real K-theory of real vector bundles, i.e., KO-theory, 
$$
KR^n(X, \tau = id) = KO^n(X)
$$
which also has the Bott periodicity 8: $KO^{n}(X) \cong KO^{n+8}(X)$.
The KO-theory of a point is given by the table
\begin{center}
 \begin{tabular}{||c| c| c| c | c | c | c | c |c |c |c ||}
 \hline
   i & 0 & 1 & 2 &  3 & 4 & 5 & 6 & 7  \\ [0.5ex]
 \hline\hline
 $KO^{-i}(pt)$ & $\mathbb{Z}$ & $\mathbb{Z}_2$ & $\mathbb{Z}_2$   & 0 & $\mathbb{Z}$  & 0  & 0 & 0 \\ [1ex]
 \hline
\end{tabular}
\end{center}

By the decomposition of the sphere $\mathbb{S}^{1,d} =  \mathbb{R}^{0,d} \cup \{ \infty \} $, one fixed point of the time reversal symmetry is $\{ 0 \} \in \mathbb{R}^{0,d}$  and
the other is $\{ \infty \}$. One can compute the KR-groups of spheres based on the above decomposition,  
\begin{equation} \label{KRSd}
  KR^{-i} (\mathbb{S}^{1,d}) = KO^{-i}(pt) \oplus KR^{-i}(\mathbb{R}^{0,d}) = KO^{-i}(pt) \oplus  KO^{-i+d}(pt)
\end{equation}
\begin{examp}
$$
KQ(\mathbb{S}^{1,2}) = KR^{-4}(\mathbb{S}^{1,2}) =  KO^{-4}(pt) \oplus  KO^{-2}(pt) =  \mathbb{Z} \oplus \mathbb{Z}_2
$$
$$
KQ(\mathbb{S}^{1,3}) = KR^{-4}(\mathbb{S}^{1,3}) =  KO^{-4}(pt) \oplus  KO^{-1}(pt) =  \mathbb{Z} \oplus \mathbb{Z}_2
$$
\end{examp}

Similarly, one can decompose $\mathbb{T}^d = (\mathbb{S}^{1,1})^d$ into fixed points plus involutive Eulidean spaces
$\mathbb{R}^{0, k}$ ($0 \leq k \leq d$), so that the KQ-groups of $\mathbb{T}^d$ are computed based on an iterative decomposition,
\begin{equation}
 KQ^{-j}(\mathbb{T}^d) = \oplus_{k =0}^d \binom{d}{k} KSp^{-j+k}(pt) = \oplus_{k =0}^d \binom{d}{k} KO^{-4-j+k}(pt)
\end{equation}
Here the $KSp$-theory $KSp(X)$ is the Grothendieck group of symplectic (or quaternionic) vector bundles over $X$, where $Sp$ indicates
its structure group is the compact symplectic group $Sp(n)$.

\begin{examp}
$$
KQ(\mathbb{T}^2) = KR^{-4}(\mathbb{T}^2) =  KO^{-4}(pt) \oplus  KO^{-2}(pt) =  \mathbb{Z} \oplus \mathbb{Z}_2
$$
$$
KQ(\mathbb{T}^3) = KR^{-4}(\mathbb{T}^3) =  KO^{-4}(pt) \oplus  3 KO^{-2}(pt)  \oplus KO^{-1}(pt)=  \mathbb{Z} \oplus 4 \mathbb{Z}_2
$$
\end{examp}
From the above example, notice that $KQ(\mathbb{T}^3)$ has $\mathbb{Z}_2$-components from both $KO^{-1}(pt)$ and $KO^{-2}(pt)$. So it is natural to
ask where does the topological $\mathbb{Z}_2$ invariant really live in, $KO^{-1}(pt)$ or $KO^{-2}(pt)$? In order to answer this question, one has to
look into the local geometry of Majorana zero modes. In the end, index theory will tell us that the topological
$\mathbb{Z}_2$ invariant lives in $KO^{-2}(pt)$, which cannot be seen from K-theory only.

Another way to compute the KQ-theory of  $\mathbb{T}^d$ is to use the
stable homotopy splitting of the torus into spheres, see \cite{FM13}.
The Baum--Connes isomorphism for the free abelian group $\mathbb{Z}^d$ provides yet another method to compute the KQ-theory of $\mathbb{T}^d$.
 The next subsection will discuss about how to compute KQ-theory using long exact sequences. 

\subsection{Long exact sequences}
\label{seqsec}
As a first upshot of our treatment, we can use several long exact sequences in K-theory, that corresponds to the different ways to introduce the $\Z_2$ invariant.

The first is the  long exact sequence in KQ-theory corresponding to the
``exact sequence'' \cite{Karoubi}[II,4.17] $Z \dashrightarrow X \dashrightarrow X\setminus Z$ for a closed $Z$. Here the dashed arrows indicate that we are looking at locally compact spaces and the morphisms in that category which are defined via their one point compactifications. The sequence reads:

\begin{equation}
 \cdots KQ^{-j-1}(Z)  \rightarrow KQ^{-j}(X \setminus Z)
\rightarrow KQ^{-j}(X) \rightarrow KQ^{-j}(Z) \rightarrow   \cdots 
\end{equation}

There exists a natural isomorphism between KR-theory and complex K-theory  given by  $KR^{-i}(Y \times \mathbb{S}^{0,1}) \cong K^{-i}(Y)$,
so by  Bott periodicity, 
$$
KQ^{-j}(Y \times \mathbb{S}^{0,1} )  = KR^{4-j }(Y \times \mathbb{S}^{0,1}) = K^{4-j}(Y) = K^{-j}(Y)
$$

\begin{lemma}
For a regular space $X$ with $Z=V^+\cap V^-$, we have the long exact sequence
\begin{equation}
   \begin{array}{ll}
     \cdots & \rightarrow K^{-j-1}(V_+^o) \rightarrow KQ^{-j-1}(X) \rightarrow KQ^{-j-1}(Z) \\
            &  \rightarrow K^{-j}(V_+^o)  \rightarrow KQ^{-j}(X) \rightarrow KQ^{-j}(Z) \rightarrow  \cdots
   \end{array}
\end{equation}
   \end{lemma}

\begin{proof}
By assumption $Z$ separates so that $X\setminus Z=V_+^o\amalg V_-^o=V_+^o\times \mathbb{S}^{0,1}$.
\end{proof}
If we are in the case of a weak $\mathbb{Z}_2$-space, one can now further decompose $Z$ iteratively.
This explains the effective  boundary used in \cite{KLW15}.

\begin{examp}
In the case of $\mathbb{T}^2$, we have $V_+^o\sim \mathbb{S}^1\times \mathbb{R}$, $Z =\mathbb{S}^{1,1} \sqcup \mathbb{S}^{1,1} $,
$$
 K(V_+^o)  \rightarrow KQ(X) \rightarrow KQ(Z)
$$
where $K(V_+^o)=K( \mathbb{S}^1\times \mathbb{R})=K^{-1}(\mathbb{S}^1)= \mathbb{Z}$, $KQ(\mathbb{T}^2) = \mathbb{Z} \oplus \mathbb{Z}_2$ and
$KQ(\mathbb{S}^{1,1}) = KQ(pt) \oplus KQ^{1}(pt) = \mathbb{Z}$. In other words, the above exact sequence gives
$$
 \mathbb{Z} \rightarrow  \mathbb{Z} \oplus \mathbb{Z}_2 \rightarrow 2\mathbb{Z}  
$$
or  in reduced K-theory,
$$
\widetilde{K}^{-1}(\mathbb{S}^1) \rightarrow  \widetilde{KQ}(\mathbb{T}^2) \rightarrow 2\widetilde{KQ}(\mathbb{S}^{1,1}), \quad \mathbb{Z} \rightarrow  \mathbb{Z}_2 \rightarrow  0
$$
This explains the reduction from $\Z$ $(= \widetilde{K}^{-1}(\mathbb{S}^1) \cong \widetilde{K}(\mathbb{T}^2) )$ to $\Z_2$ $(= \widetilde{KQ}(\mathbb{T}^2))$,
that is, the $\mathbb{Z}_2$ invariant  results from the time reversal symmetry while the complex K-theory becomes the Quaternionic K-theory.

\end{examp}

When $X$ is a weak $\mathbb{Z}_2$-space,  namely, there exists a decomposition $X=X_+\amalg X_-\amalg X^\tau$,  one has a long exact sequence involving  different K-theories.

\begin{lemma} If  $X$ is a weak $\mathbb{Z}_2$-space,
   then there exists a long exact sequence
   \begin{equation}
   \begin{array}{ll}
     \cdots & \rightarrow KQ^{-j-1}(X_+\cup X_-) \rightarrow KQ^{-j-1}(X) \rightarrow KSp^{-j-1}(X^\tau) \\
            &  \rightarrow KQ^{-j}(X_+\cup X_-)  \rightarrow KQ^{-j}(X) \rightarrow KSp^{-j}(X^\tau) \rightarrow  \cdots
   \end{array}
      \end{equation}  where $X_+\cup X_-$ has a free $\Z_2$ action interchanging the spaces.
If $X_+$ and $X_-$ are in different components of $X\setminus X^\tau$,
   then the above long exact sequence is reduced to
\begin{equation}
   \begin{array}{ll}
     \cdots & \rightarrow K^{-j-1}(X_+) \rightarrow KQ^{-j-1}(X) \rightarrow KSp^{-j-1}(X^\tau) \\
            &  \rightarrow K^{-j}(X_+)  \rightarrow KQ^{-j}(X) \rightarrow KSp^{-j}(X^\tau) \rightarrow  \cdots
   \end{array}
\end{equation}

\end{lemma}
\begin{proof}
Using $X^{\tau} \dashrightarrow X \dashrightarrow X\setminus X^{\tau}$, we get the first sequence by noticing that
when restricting the involution $\tau$ to the fixed points,
$\tau|_{X^\tau} $ becomes trivial, so $KQ^{-j}(X^\tau, \tau|_{X^\tau}) = KSp^{-j}(X^\tau)$.
Now under the assumption that  $X_+$ and $X_-$ are in different components of $X\setminus X^\tau$, we have  $X \setminus X^\tau = X_+\amalg X_-=X_+ \amalg  \tau(X_+)
= X_+ \times \mathbb{S}^{0,1}$, it follows that $ KQ^{-j}(X \setminus X^\tau )=
KQ^{-j}(X_+ \times \mathbb{S}^{0,1})\simeq K^{-j}(X_+)$.

\end{proof}
We can also replace the middle terms by $KQ^{-j}(X)\simeq KR^{-j+4}(X)$.
This sequence is at the heart of the description of \cite{FKM07}.


\begin{examp} When $X = \mathbb{S}^{1,3}$, $ X^\tau =\mathbb{S}^{0,1}$ and the open set
$ \mathbb{S}^{1,3} \setminus \mathbb{S}^{0,1} =  \mathbb{R}^{0,3} \setminus \{ 0 \} = X_+ \cup X_- $, where
$X_+ \sim 3\mathbb{R} \cup 6 \mathbb{R}^2 \cup 4 \mathbb{R}^3$. We extract two parts from the long exact sequence, the first part is
   $$
  KSp^{-6}(\mathbb{S}^{0,1}) \rightarrow  {K}^{-5}(X_+) \rightarrow {KQ}^{-5}(\mathbb{S}^{1,3}) \rightarrow {KSp}^{-5}(\mathbb{S}^{0,1})
   $$
that is,
  $$
   2 \mathbb{Z}_2 \rightarrow  3 \mathbb{Z} \oplus 4 \mathbb{Z} \rightarrow  \mathbb{Z}_2 \rightarrow 2 \mathbb{Z}_2
   $$
And the second part is
 $$
  KSp^{-2}(\mathbb{S}^{0,1}) \rightarrow  {K}^{-1}(X_+) \rightarrow {KQ}^{-1}(\mathbb{S}^{1,3}) \rightarrow {KSp}^{-1}(\mathbb{S}^{0,1})
   $$
 that is,
  $$
   0 \rightarrow  3 \mathbb{Z} \oplus 4 \mathbb{Z} \rightarrow  \mathbb{Z}_2 \rightarrow 0
   $$


\end{examp}

Another sequence is the Mayer--Vietoris sequence for a regular space $X$
covered by two fundamental domains, i.e., $X = V_+\cup V_-$,
\begin{equation}
  \begin{array}{ll}
    & \cdots  \rightarrow KQ^{-j-1}(V_+\cap V_-)\stackrel{\Delta}{\rightarrow}  \\
   KQ^{-j}(X)
\stackrel{u}{\rightarrow} KQ^{-j}(V_+)\oplus KQ^{-j}(V_-)
         & \stackrel{v}{\rightarrow}  KQ^{-j}(V_+\cap V_-) \rightarrow   \cdots 
  \end{array}
\end{equation}
where $u(\alpha)=(\alpha|_{V+},\alpha|_{V_-})$ and
$v(\alpha_1,\alpha_2)=\alpha_1|_{V_+\cap V_-}-\alpha_2|_{V_+\cap V_-}$.

\begin{examp} \label{keyex}
This is the long exact sequence at the heart of the argument of \cite{MB07}.
When $X=\mathbb{T}^2$, $V_+=V_-=C = \mathbb{S}^{1,1} \times \mathbb{R}^{0,1}$, where $C$ is an open cylinder with the time reversal $\mathbb{Z}_2$-action 
and $Z = V_+\cap V_-=\mathbb{S}^{1,1}\amalg \mathbb{S}^{1,1}$.
$$
  2KQ^{-1}(\mathbb{S}^{1,1}) \rightarrow  {KQ}(\mathbb{T}^2) \rightarrow 2{KQ}(C) \rightarrow 2{KQ}(\mathbb{S}^{1,1})
   $$
Since $KQ(\mathbb{T}^2) \cong KQ(\mathbb{S}^{1,2})$, so the above sequence is the same as
   $$
  2KQ^{-1}(\mathbb{S}^{1,1}) \rightarrow  {KQ}(\mathbb{S}^{1,2}) \rightarrow 2{KQ}(C) \rightarrow 2{KQ}(\mathbb{S}^{1,1})
   $$
  The KQ-theory of the cylinder $C$ is,
   $$
   KQ(C) = KR^{-4}(\mathbb{S}^{1,1} \times \mathbb{R}^{0,1}) = KR^{-3}(\mathbb{S}^{1,1}) = KO^{-3}(pt) \oplus KO^{-2}(pt) = \mathbb{Z}_2
   $$
  so that,
    $$
    \mathbb{Z} \oplus  \mathbb{Z} \rightarrow    \mathbb{Z} \oplus  \mathbb{Z}_2 \rightarrow  \mathbb{Z}_2 \oplus  \mathbb{Z}_2  \rightarrow  \mathbb{Z} \oplus  \mathbb{Z}
   $$
  Hence the $\mathbb{Z}_2$ invariant living in $\widetilde{KQ}(\mathbb{S}^{1,2})$ can also be found in $ KQ(C)$.
\end{examp}

   Finally there is the relative sequence, for a closed subspace $Y$ in $X$

   \begin{equation}
 KQ^{-j-1}(X,Y)\to KQ^{-j}(X) \to KQ^{-j}(Y)\to KQ^{-j}(X,Y)
   \end{equation}
   if one identifies $KQ^{-n}(X,Y)$ with $\widetilde{KQ}(S^n(X/Y)))$, where $S^n$ is the n-fold suspension,
   then the first map is induced by the quotient $q:X\to X/Y$.

   \begin{examp}
  Using this for $\mathbb{T}^2$ and $\mathbb{S}^{1,1}\vee\mathbb{S}^{1,1}$,
  one obtains the collapse map $q:\mathbb{T}^2\to \mathbb{S}^2$ and similarly
  for $\mathbb{T}^3$. This is what is used in \cite{FM13}.
   \end{examp}

\subsection{Baum--Connes Isomorphism} \label{BCIso}

In \cite{K09}, Kitaev mentioned the real Baum--Connes conjecture \cite{BK04}, which is true for the abelian free group $\Gamma = \mathbb{Z}^d$,
i.e., the translational symmetry group,
so we call it the Baum--Connes isomorphism in this paper. The assembly map
was used to understand and calculate the  KO-theory of $\mathbb{T}^d$   in \cite{K09}. In this subsection we briefly review the Baum--Connes
isomorphism for $\Gamma = \mathbb{Z}^d$, which will be useful for the bulk-boundary correspondence in a later section.

\begin{defn}
   Let $\Gamma$ be a  discrete countable group, the assembly map $\mu$ is a morphism from the equivariant
    K-homology of  the classifying space of  proper actions $E\Gamma$   to the
    K-theory of the reduced group $C^*$-algebra of $\Gamma$, i.e.,
   \begin{equation*}
      \mu(\Gamma): K_j^\Gamma(E\Gamma) \rightarrow K_j(C_\lambda^*(\Gamma, \mathbb{C}))
   \end{equation*}
\end{defn}

The classical complex Baum--Connes conjecture states that the assembly index map $\mu$ is an isomorphism. If $\Gamma$ is torsion free, then the left hand side
is reduced to the  K-homology of the ordinary classifying space $B\Gamma$, i.e., $K^\Gamma_j(E\Gamma) = K_j(B\Gamma) $.

\begin{defn}
  In the real case, the assembly map is similarly defined as
  \begin{equation*}
      \mu_\mathbb{R}(\Gamma): KO_j^\Gamma(E\Gamma) \rightarrow KO_j(C_\lambda^*(\Gamma, \mathbb{R}))
   \end{equation*}
\end{defn}

The real Baum--Connes conjecture follows from the complex Baum--Connes conjecture, so $ \mu_\mathbb{R}(\mathbb{Z}^d)$ is an isomorphism.

If one defines the real  function algebra on the Real space $(X, \tau)$,
$$C_0(X, \tau):= \{f \in C_0 (X ) \, |\,  \overline{f (x)} = f (\tau (x))\}$$
then the KR-theory of $(X, \tau)$ is identified with the topological K-theory of the above  real function algebra \cite{R14},
$$
KR^{ -i} (X, \tau ) = KO_i (  C_0 (X, \tau )  )
$$

By the real assembly map $ \mu_\mathbb{R}(\mathbb{Z}^d)$, one has the Baum--Connes isomorphism, sometimes  also called the dual Dirac isomorphism,
since $\mathbb{T}^d$ is the classifying space for $\mathbb{Z}^d$ with the universal cover $\mathbb{R}^d$,
\begin{equation}
\label{bceq}
       KO_i(\mathbb{T}^d) = KO_i^{\mathbb{Z}^d} (\mathbb{R}^d) \cong KO_i(C^*(\mathbb{Z}^d, \mathbb{R})) = KO_i (  C (\mathbb{T}^d, \tau )  ) 
\end{equation}
By the Poincar\'{e} duality, the real K-homology on the left hand side is identified with a KO-group,
i.e., $KO_i(\mathbb{T}^d)   = KO^{d-i}(\mathbb{T}^d)$. Thus the KQ-groups of $\mathbb{T}^d$ can be computed by the relevant $KO$-groups,
\begin{equation}
  KQ^{-i}( \mathbb{T}^d) = KR^{ -i-4 }( \mathbb{T}^d )= KO_{i + 4}(\mathbb{T}^d)= KO^{d-i-4}(\mathbb{T}^d)
\end{equation}

\begin{examp}
$$
KQ(\mathbb{T}^2) = KR^{-4}(\mathbb{T}^2) =  KO^{-2}(\mathbb{T}^2)  
$$
$$
KQ(\mathbb{T}^3) = KR^{-4}(\mathbb{T}^3) =  KO^{-1}(\mathbb{T}^3)
$$
$$
KQ^{-1}(\mathbb{T}^3) = KR^{-5}(\mathbb{T}^3) =  KO^{-2}(\mathbb{T}^3)
$$
\end{examp}

\section{Analytical index}\label{sec:Aind}

In this section,  we will look into the local geometry of Majorana zero modes and the relevant analytical index.
First we will recall some basic facts about $KR$-cycles in spin geometry \cite{LM90, VFG01}. After that we will give the
definition of Majorana zero modes,   and the topological $\mathbb{Z}_2$ invariant can be defined as the parity of Majorana zero modes.
By localization, we will interpret the topological $\mathbb{Z}_2$ invariant  as the mod 2 analytical index of the effective Hamiltonian.
A KQ-cycle will be defined to model a Kramers pair, and a localized Majorana zero mode is a coupled product of two KR-cycles. 

\subsection{KR-cycles}
\label{cyclesec}
In order to understand Majorana zero modes in a   topological insulator,
we compare them to Majorana spinors.
In spin geometry,  one models  spinors by $KR$-cycles in the framework of $KR$-homology, which is the dual theory of KR-theory.

\begin{defn}
  A K-cycle for a $*$-algebra of operators $\mathcal{A}$ is a triple $(\mathcal{A}, \H, D)$,
  commonly called a spectral triple in noncommutative geometry, where $\H$ is a complex Hilbert space, and $\mathcal{A}$ has a faithful $*$-representation on $\H$ as bounded operators,
  i.e., $\pi : \mathcal{A} \rightarrow B(\H)$. $D$ is a self-adjoint (typical unbounded) operator  with compact resolvent such that
  the commutators $[D, \pi(a)]$ are bounded operators for all $a \in \mathcal{A}$.
\end{defn}

In practice, $\mathcal{A}$ is always assumed to be a unital (pre-)$C^*$-algebra. When the representation $\pi$ is understood, it is always skipped from the notation.
$D$ is usually taken as a self-adjoint Dirac-type operator, and there is a natural action by the real Clifford algebra $Cl_{p,q}$.

A  K-cycle $(\mathcal{A}, \H, D)$ is even (or graded) if there exists a  grading operator $\gamma$
with $\gamma^* = \gamma$ and $\gamma^2 = 1$ such that $D\gamma = -\gamma D$ and $\gamma a = a\gamma$ for all $a \in \mathcal{A}$.
Otherwise, a K-cycle is odd (or ungraded). If there exists a grading $\gamma$,
then the Hilbert space $\H$ is also assumed to be $\mathbb{Z}_2$-graded. 

A (general) real structure on an even K-cycle is defined by an anti-linear isometry $J$ such that
\begin{equation}\label{RealCond}
  JD = DJ, \quad J^2 = \epsilon, \quad J\gamma = \varepsilon \gamma J
\end{equation}
where the signs $\epsilon, \varepsilon = \pm 1$  depend on  $2k$ mod 8. Real structures on a spectral triple was first introduced by Connes \cite{C95}. 

\begin{defn}
  A $KR_{2k}$-cycle (depending on  $2k$ mod $8$) is defined as an even K-cycle equipped with a real structure, i.e., a quintuple $( \mathcal{A}, \H, D, J, \gamma)$,
  satisfying the relations
  $$
   JD = DJ, \quad J^2 = \epsilon, \quad J\gamma = (-1)^k \gamma J
  $$
  where $J^2 = 1$ if $2k = 0, 6$ mod $8$ and $J^2 = -1$ if $2k = 2, 4$ mod $8$.
\end{defn}

Here we give the representation theory of spinors in spin geometry \cite{LM90}.
By the representation theory of the real Clifford algebra $Cl_{2k, 0}$, when $2k = 0$ mod 8, it has a unique real pinor representation
(i.e., Majorana pinor) and there are two inequivalent real spinor representations (i.e., Majorana--Weyl spinors).
When $2k =2$ mod 8, it has a unique quaternionic pinor representation
(i.e., symplectic Majorana pinor) and there are two inequivalent complex spinor representations (i.e.,  Majorana--Weyl spinors).
When $2k =4$ mod 8, it has a unique quaternionic pinor representation
(i.e., symplectic Majorana pinor) and there are two inequivalent quaternionic spinor representations (i.e.,  symplectic Majorana--Weyl spinors).
When $2k = 6$ mod 8, it has a unique real pinor representation
(i.e., Majorana pinor) and there are two inequivalent complex spinor representations (i.e., Majorana--Weyl spinors).

\begin{examp}
   When $2k = 2$ mod $8$, one has  Majorana--Weyl spinors
   modeled by a $KR_{2}$-cycle $( \mathcal{A}, \H, D, J, \gamma)$ such that
   \begin{equation*}
  JD = DJ, \quad J^2 = -1, \quad J\gamma = -\gamma J
  \end{equation*}
\end{examp}

When there is no grading operators, one defines an odd $KR$-cycle by a quadruple $( \mathcal{A}, \H, D, J)$.
\begin{defn}
   A $KR_{2k-1}$-cycle (for $2k-1 = 1, 3, 5, 7 \mod 8$) is  defined as a quadruple $( \mathcal{A}, \H, D, J)$ satisfying the relations
   $$
    JD = (-1)^{k} DJ, \quad J^2= \epsilon
   $$
   where $J^2 = 1$ if $2k-1 = 1, 7$ mod $8$ and $J^2 = -1$ if $2k-1 = 3, 5$ mod $8$.
\end{defn}

The representation theory of the real Clifford algebra in odd dimensions is easier, when $2k-1 = 1, 7$ mod $8$, there is a unique real spinor representation;
when $2k-1 = 3, 5$ mod $8$, there is a unique quaternionic spinor representation.

\begin{examp}
   When $2k-1 = 5$ mod $8$, one has a quaternionic spinor modeled by   a $KR_{5}$-cycle $( \mathcal{A}, \H, D, J)$ such that
   $$
   JD = -DJ, \quad J^2= -1
   $$
\end{examp}

From a K-cycle $(\mathcal{A}, \H, D)$, one obtains the corresponding Fredholm module  $(\mathcal{A}, \H, F)$ by setting $F = D (1+D^2)^{-1/2}$.
By definition, the set of equivalence classes of Fredholm modules modulo unitary equivalence and homotopy equivalence defines the K-homology group, for details see \cite{BHS07, HR00}.
The following example is the classical Dirac geometry modeling Majorana spinors in spin geometry \cite{LM90, VFG01}.

\begin{examp}\label{DiracGeo}
   Let $M$ be a compact spin manifold of dimension $2k$, its Dirac geometry is defined as the unbounded K-cycle $(C^\infty(M), L^2(M, \slashed{S}), D\!\!\!\!/ \,) $,  where
   $L^2(M, \slashed{S})$ is the Hilbert space of spinors and $D\!\!\!\!/$ is the Dirac operator. The grading operator $\gamma$, or $c(\gamma)$ for the Clifford multiplication by $\gamma$,
   is defined as usual in an even dimensional Clifford algebra. In addition, the canonical real structure is given
   by the charge conjugation operator $C$ acting on the Clifford algebra. Thus the quintuple $(C^\infty(M), L^2(M, \slashed{S}), D\!\!\!\!/ \, , C, \gamma) $ defines a
   $KR_{2k}$-cycle of spinors. In spin geometry, when a spinor $\psi \in L^2(M, \slashed{S})$ satisfies the real condition $C \psi = \psi$ (generalizing $\psi^\dagger = \psi$ used in physics), 
   it is called a Majorana spinor, the space of Majorana spinors
   is denoted by $ L^2(M, \slashed{S}; C)$.

\end{examp}

As a summary, a  KR-cycle   satisfies the commutation relations given in the following table 
\begin{center}
 \begin{tabular}{||c| c| c| c | c | c | c | c |c |c |c ||}
 \hline
   $ j\,\, \text{mod 8}$ & 0 & 1 & 2 &  3 & 4 & 5 & 6 & 7  \\ [0.5ex]
 \hline\hline
  $J^2 = \pm 1$ & $+$ & $+$ & $-$   &  $-$  &  $-$  &  $-$  & $+$ & $+$ \\
 \hline
  $JD = \pm DJ $ & $+$ & $-$ & $+$   &  $+$  &  $+$  &  $-$  & $+$ & $+$ \\  
 \hline
 $J\gamma = \pm \gamma J $ & $+$ &   & $-$   &     &  $+$  &     & $-$ &   \\ [1ex]
 \hline
\end{tabular}
\end{center} 
In addition, a (reduced) $KR_j$-cycle is said to have the  KO-dimension $j$ mod $8$. KO-dimension of a KR-cycle, first introduced by Connes in \cite{C95}, coincides with
 the degree of the corresponding KR-homology. The idea comes from the KO-orientation represented by a fundamental class in KO-homology, 
which induces the Poincar{\'e} duality between KO-homology and KO-theory.

\subsection{Majorana zero modes}

We introduce the notion of (localized) Majorana zero modes,  the parity of  Majorana
zero modes is a characterization of the topological $\mathbb{Z}_2$ invariant. 
Analogous to those in Bogoliubov--de Gennes (BdG) topological superconductors,  
 Majorana zero modes in topological insulators can be viewed as Bogoliubov quasi-particles. 
Due to the characteristic of quasi-particles, Majorana zero modes has a new geometry compared to Majorana spinors. 

Let $H(x)$ be a single-particle Hamiltonian  parametrized by the points in the momentum space $X$,  $H$ is assumed to be time reversal invariant, i.e., 
$$
\Theta H(x) \Theta^{*} = H(\tau(x)), \quad \forall \,\, x \in X
$$
where $\Theta$ is the  time reversal operator.
 In topological insulators, the Hamiltonian $H(x)$ can be effectively approximated by a Dirac operator plus quadratic correction terms around a fixed point,
 i.e., $H(x) \sim D(x) + O(x^2)$ for $x \in U(x_0), x_0 \in X^\tau$, see modified Dirac Hamiltonians in \cite{KLW15, S13}.
 In the mathematical physics literature, some authors use the Dirac operator as the effective Hamiltonian and the starting point, for example see \cite{GS15}. 

We assume the Hilbert bundle $\pi: \mathcal{H} \rightarrow X$ has rank two, and the Hilbert space is generated by a pair of global sections $\phi$ and $\Theta \phi$ 
(i.e., a Kramers pair consisting of an electron and its time reversal partner).
Due to the Kramers degeneracy, the Kramers pair $( \phi ,  \Theta  \phi  )$ has the same energy (or lives in the same band). 
We consider the quasi-particle made up of $\phi$ and $\Theta  \phi$, and define an effective Hamiltonian of this free fermionic system.
\begin{defn}
   The effective Hamiltonian $\tilde{H}$ of a   topological insulator is defined by
   \begin{equation} \label{effH}
      \tilde{H}(x): = \begin{pmatrix}
                     0 & \Theta H(x) \Theta^* \\
                     H(x) & 0
                  \end{pmatrix}
   \end{equation}
   acting on a Kramers pair $( \phi ,  \Theta  \phi  ) \in \Gamma(X, \mathcal{H})$.
\end{defn}

  Our definition of the effective Hamiltonian is different from the convention used in physics, where it is   diagonal such as $  \begin{pmatrix}
                                                                 H_\uparrow & 0 \\
                                                                 0 & H_\downarrow
                                                                \end{pmatrix}$ 
                                                                with $ H_\uparrow$, $H_\downarrow$ modeling electronic particles with opposite spins  (or chirality), see examples in \cite{KLW15}.                                                                                            
   One reason behind our definition is that we have to point out a Kramers pair describes a \emph{coupled} quasi-particle, 
   another important reason is that we follow the mathematical convention used in spin geometry for spinors and the  Dirac operator, see below.   
   We will see that the localization of the effective Hamiltonian gives rise to a skew-adjoint operator, whose analytical index is $\mathbb{Z}_2$-valued.
   Notice that the effective Hamiltonian $\tilde{H}$ is also time reversal invariant, 
   \begin{equation*}
\Theta \tilde{H}(x) \Theta^{*} = \begin{pmatrix}
                                   0 & \Theta H(\tau(x)) \Theta^* \\
                                   \Theta H(x) \Theta^* & 0
                                 \end{pmatrix} = \begin{pmatrix}
                                                    0 & H(x) \\
                                                    H(\tau(x)) & 0
                                                 \end{pmatrix} = \tilde{H}(\tau(x)) 
\end{equation*} 

\begin{lemma}

  If $\phi$ is an eigenstate of $H$ with eigenvalue $E$, then the eigenvalue equation of $\tilde{H}$ has  the following matrix form,
    \begin{equation} \label{EigEq}
    \begin{pmatrix}
        0 & \Theta H(x) \Theta^* \\
        H(x) & 0
      \end{pmatrix}
   \begin{pmatrix}
       \phi(x) \\
       \Theta \phi (x)
   \end{pmatrix}
   =   \begin{pmatrix}
       E(\tau(x)) \, \Theta \phi(x) \\
       E(x) \, \phi(x)
   \end{pmatrix}
\end{equation}

\end{lemma}

In spin geometry, the Dirac operator acting on a Dirac spinor (with two components) is always decomposed as $D = \begin{pmatrix}
                                                              0 & D_- \\
                                                              D_+ & 0
                                                            \end{pmatrix}$,
where $D_+$ (resp. $D_-$) corresponds to the eigenvalue $+1$ (resp. $-1$) of the grading operator. 
In addition, $D_\pm: \mathcal{H}_\pm \rightarrow \mathcal{H}_\mp$ interchanges the Hilbert space of  spinors with opposite spins. 
Similarly, $\tilde{H}$ is constructed  off-diagonally since a topological insulator is a fermionic chiral system and time reversal symmetry changes chirality. 
Based on the  eigenvalue equation  \eqref{EigEq},
the effective Hamiltonian $\tilde{H}$ plays a similar role as the $\mathbb{Z}_2$-graded Dirac operator $D$.
In this analogy,  $H$ (resp. $\Theta H \Theta^*$) is compared to $D_+$ (resp. $D_- = D_+^*$),
and the real structure is defined by $\Theta$ instead of the $*$-operation. Notice that the chirality of a pair of chiral states is switched after applying $\tilde{H}$ (indeed $\Theta$).
We have to point out
the big difference between the $\mathbb{Z}_2$-graded Dirac operator $D$ and $\tilde{H}$ is that $D$ is a self-adjoint operator 
but $\tilde{H}$ can be approximated by a skew-adjoint operator via localization.

By assumption, the single-particle Hamiltonian $H$ is self-adjoint, i.e., $H^*(x) = H(x)$, so  we  have 
\begin{equation}
\tilde{H}^*(x) = \begin{pmatrix}
                   0 & H(\tau(x)) \\
                   H(x) & 0
                 \end{pmatrix}^* = \begin{pmatrix}
                                     0 & H(x) \\
                                     H(\tau(x)) & 0
                                   \end{pmatrix} = \tilde{H}(\tau(x))
\end{equation}
which means the effective Hamiltonian $\tilde{H}$ is neither self-adjoint nor skew-adjoint. 
In order to construct a skew-adjoint operator and the resulting mod 2 analytical index, a localization process will be applied to the raw data ${H}$.
As a remark, we have to stress that the real structures defined by the adjoint  $*$-operation and the time reversal operator $\Theta$ are totally different.

\begin{defn}
  A Majorana state with respect to a given real structure $J$ is defined as a state $\psi$  
  satisfying the real condition ${J} \psi = \psi$.
\end{defn}

 \begin{rmk}
   Our definition of Majorana states is a generalized version of the definition used in physics.
   For example from \cite{DFN15,W09}, Majorana states (or Majorana zero modes) are defined as self-adjoint fermionic operators commuting with the Hamiltonian (in the CAR algebra).
   In our definition, we follow the conventions  in the Dirac geometry using the charge conjugation, see Example \ref{DiracGeo}. 
   So  the dagger $\dagger$-operation  (commonly used in physics) is
   replaced by a real structure $J$, and in order for a state $\psi$ to be Majorana (or real), the   condition   $\psi^\dagger = \psi$ is generalized to the real condition
   $J\psi = \psi$.

\end{rmk} 

 Based on the time reversal operator $\Theta$, we define a new real structure by
$$
\mathcal{J} := \begin{pmatrix}
                                                      0 & \Theta^* \\
                                                      \Theta & 0
                                                    \end{pmatrix}
$$ such that $ \mathcal{J}^* = \mathcal{J}$ and $\mathcal{J}^2 = 1$.
Since the real structure $\mathcal{J}$ acts on a Kramers pair $\Phi = (\phi, \Theta \phi)$ as
$$
 \begin{pmatrix} 0 & \Theta^* \\
                              \Theta & 0
              \end{pmatrix}
              \begin{pmatrix}
                 \phi \\
                 \Theta \phi
              \end{pmatrix}
              = \begin{pmatrix}
                 \phi \\
                 \Theta \phi
              \end{pmatrix} \quad \text{i.e.}\,\, \mathcal{J} \Phi = \Phi
$$
such a quasi-particle state $\Phi = (\phi, \Theta \phi)$ is real with respect to $\mathcal{J}$.
So a Kramers pair gives the canonical Majorana state in topological insulators.

Our notation of Majorana states can be translated to the familiar Majorana operators commonly used in physics, and vice versa.
Suppose $\gamma_1$ and $\gamma_2$ are Majorana operators,  $f_+$ and $f_-$ are fermionic (creation and  annihilation) operators related to $\gamma_i$ by 
\begin{equation*}
 f_+ = \gamma_1 + i\gamma_2, \quad f_- = \gamma_1 - i\gamma_2
\end{equation*}
Or the Majorana operators can be expressed in fermionic operators,
\begin{equation*}
 \gamma_1 = \frac{1}{2} (f_+ + f_-), \quad  \gamma_2 = \frac{1}{2i} (f_+ - f_-)
\end{equation*}
There is a state-operator correspondence involved in the discussion here. 

By a Bogoliubov transformation, a Kramers pair  $\Phi = (\phi, \Theta \phi)$  can be  rewritten as
\begin{equation}
 \begin{pmatrix} 
 \gamma_1 \\
 \gamma_2
 \end{pmatrix} = \begin{pmatrix}
                 1 & i \\
                 -i & 1
                \end{pmatrix} \begin{pmatrix}
                               \phi \\
                               \Theta \phi
                               \end{pmatrix} = \begin{pmatrix}
                                                \phi + i \Theta \phi \\
                                                -i \phi + \Theta \phi
                                                \end{pmatrix}
\end{equation}
The matrix $B = \begin{pmatrix}
                 1 & i \\
                 -i & 1 
                \end{pmatrix}$ induces  the Bogoliubov transformation, in this context, the Majorana states $\gamma_i$ are also called Bogoliubov quasi-particles.
 The fermionic states can be recovered as 
 \begin{equation}
   \tilde{f}_+ = \gamma_1 + \Theta \gamma_2 = 2i \Theta \phi, \quad \tilde{f}_- = \gamma_1 - \Theta \gamma_2 = 2 \phi                
 \end{equation}
 Notice that the real structure used here is  $\Theta$ instead of $i$ in the original case. 
 The Bogoliubov quasi-particles are equivalently obtained by fermionic states,
 \begin{equation}
    \gamma_1 = \frac{1}{2}(\tilde{f}_+ + \tilde{f}_-), \quad  \gamma_2 = -\frac{\Theta}{2}(\tilde{f}_+ - \tilde{f}_-)
 \end{equation}
 The imaginary unit $i$ in $\gamma_i$ is a formal symbol to connect two particles and form a quasi-particle. 
 In order to avoid unnecessary confusions   introduced by $i$ (as another real structure), we tend to use the vector form  $\Phi = (\phi, \Theta \phi)$ to denote a Majorana state. 
 
 If a Kramers pair $(\phi, \Theta \phi)$ is written as a Bogoliubov quasi-particle $\gamma_1 = \phi + i\Theta \phi$, then the effective Hamiltonian $\tilde{H} = \begin{pmatrix}
                                                                                                                                                       0 & \Theta H \Theta^* \\
                                                                                                                                                       H & 0
                                                                                                                                                      \end{pmatrix}$ 
                           is accordingly written as $H_1 = H + i \Theta H \Theta^*$. 
 Similarly, for the other Bogoliubov quasi-particle $\gamma_2 = \Theta \phi -i \phi  $, the corresponding Hamiltonian is $H_2 = \Theta H_1 \Theta^* = \Theta H \Theta^* - iH$. 
 The action of $\Theta$ maps $(\phi, \Theta \phi)$ to $(\Theta \phi, -\phi)$, so $\gamma_2$ is the mirror image of $\gamma_1$ under the time reversal symmetry, i.e., $\gamma_2 = \Theta \gamma_1$. 
  Define a new Hamiltonian acting on Bogoliubov quasi-particles $(\gamma_1, \gamma_2) = (\gamma_1, \Theta \gamma_1)$,
 \begin{equation}
 \hat{H} := \begin{pmatrix} 
            0 & H_2 \\
            H_1 & 0
           \end{pmatrix} = \begin{pmatrix}
                           0 & \Theta H_1 \Theta^* \\
                           H_1 & 0
                         \end{pmatrix}
                           = \begin{pmatrix}
                            0 &  \Theta H \Theta^* - iH \\
                           H + i \Theta H \Theta^* & 0 
                           \end{pmatrix}  
 \end{equation}
  At the first glance, the analytical index of $\hat{H}$ is 
  \begin{equation}
    ind_a(\hat{H}) = dim_\mathbb{H} ker ( H + i \Theta H \Theta^*) - dim_\mathbb{H} ker (\Theta H \Theta^* - iH) = 0
  \end{equation}
  which is always zero since $\Theta H \Theta^* - iH = -i(H + i \Theta H \Theta^* ) $.  
  Therefore, the non-trivial invariant  is expected to be defined as a mod 2  index,
  \begin{equation}
     ind_a(\hat{H}) := dim_\mathbb{H} ker ( H + i \Theta H \Theta^*)  \,\, \text{ mod 2} 
  \end{equation}
  which is the parity of zero modes of Bogoliubov quasi-particles.
  This heuristic reasoning leads to the desired definition of the topological $\mathbb{Z}_2$ invariant, but we still need a skew-adjoint operator to define a mod 2 analytical index.

In a Kramers pair $\Phi = (\phi, \Theta \phi)$ (as a Majorana state), the chiral states $\phi$ and $\Theta \phi $  are orthogonal, hence linearly independent.
 The chiral states could intersect with each other, i.e.,  $\Theta \phi (x) = \phi(x)$, at a   fixed point $x\in X^\tau$. In particular,
a localized Majorana zero mode $\Phi_0 = (\phi_0, \Theta \phi_0)$ must have zero energy at that fixed point,
$$
 \Phi_0(x) = \textbf{0} \quad or \quad \phi_0(x)=\Theta \phi_0(x)=0
$$
Recall that the zero energy level is defined by the Fermi level, which is assumed to fall into the band gap. 
Strictly speaking, a zero mode $(\phi_0, \Theta \phi_0)$ must be a pair of edge states passing through the Fermi level.
Since there exists a one-to-one correspondence between Majorana (bulk) states and their edge states by partial Fourier transformations  \cite{KLW15}, 
we abuse the notation and still use  Majorana states $(\phi_0, \Theta \phi_0)$, otherwise
we can always adjust the zero energy level to make it work.

Let us look at a pair of localized Majorana zero modes $(\gamma_1^0, \gamma_2^0 = \Theta \gamma^0_1)$ around a fixed point,
\begin{equation*}
 \gamma_1^0 = \begin{pmatrix}
                  \phi_0 \\
                  \Theta \phi_0
                 \end{pmatrix}, \quad
     \gamma_2^0 = \begin{pmatrix}
                \Theta \phi_0 \\
                - \phi_0 
                \end{pmatrix}
\end{equation*}
At this fixed point, a winding matrix (or a transition function of the Hilbert bundle) takes one of the following form
\begin{equation*}
  w_1 = \begin{pmatrix}
         0 & -1 \\
         1 & 0
        \end{pmatrix} \quad or \quad  w_2 = \begin{pmatrix}
         0 & 1 \\
         -1 & 0
        \end{pmatrix}
\end{equation*}
A direct matrix-vector multiplication shows that 
\begin{equation*}
  w_1 \gamma_1^0 = -\gamma_2^0, \,\, w_2 \gamma_1^0 = \gamma_2^0; \quad w_1\gamma_2^0 = \gamma_1^0, \,\, w_2\gamma_2^0 = - \gamma_1^0
\end{equation*}
In practice, we mainly focus on the real part  $\Phi_0^+ = (\phi_0, \Theta \phi_0) = \gamma_1^0$, and the imaginary part  (with respect to $\Theta$)
 $\Phi_0^- = ( \Theta \phi_0, - \phi_0) = \gamma_2^0$ is skipped implicitly but can be easily restored.
 Inspired by spin geometry, we are actually interested in a Majorana state that changes sign, i.e. $\gamma_1^0  \rightarrow - \gamma_2^0$, 
 so we define it as a (generalized) Majorana zero mode. We stress that a Majorana state that remains the same sign, i.e., $\gamma_1^0  \rightarrow \gamma_2^0$, will not be counted as a Majorana zero mode.

\begin{defn}
   A  (generalized) Majorana zero mode is defined as a localized Majorana state $\Phi_0 = (\phi_0, \Theta \phi_0)$ in a small neighborhood of a fixed point $x \in X^\tau$ so that $\phi_0(x) = \Theta \phi_0(x) = 0 $
  and changes sign  at that fixed point.

\end{defn}
By definition, a localized Majorana zero mode could be found only in a small neighborhood of a fixed point.
The local picture of a Majorana zero mode is given by a conical singularity, for example a cone $V(x^2 +y^2 - z^2)$ in 3d around the intersection,
such cone is called a Dirac cone and the intersection point is called a Dirac point by physicists.

\begin{defn}
  The  topological $\mathbb{Z}_2$ invariant  of a time reversal invariant topological insulator is defined as the parity of Majorana zero modes.
\end{defn}
This practical definition gives the physical meaning of the topological $\mathbb{Z}_2$ invariant,
we will interpret it as a mod 2 analytical index  in the next subsection.

\subsection{Mod 2 index}

Atiyah and Singer introduced a mod $2$ analytical index of real skew-adjoint elliptic operators in   \cite{AS71},
$$
 ind_a (P) := \dim_{\mathbb{R}} \ker P \quad \text{mod 2}
$$
In this subsection, we will explain that the  topological $\mathbb{Z}_2 $ invariant can be interpreted as the
mod 2  analytical index of the effective Hamiltonian $\tilde{H}$.
The idea of spectral flow can be used to compute the analytical index.  

Recall  the effective Hamiltonian  is defined by
$
  \tilde{H}(x) = \begin{pmatrix}
        0 & \Theta H(x) \Theta^* \\
        H(x) & 0
      \end{pmatrix}
$, where $\Theta H(x) \Theta^* = H(\tau(x))$ and $H^*(x) = H(x)$ for all $x \in X$.
Since the set of fixed points is assumed to be  finite, we always assume $H$ is a Fredholm operator, so is $\tilde{H}$.

\begin{lemma}
  The effective Hamiltonian $\tilde {H}$ can be approximated by a skew-adjoint operator by localization.
\end{lemma}
\begin{proof}
In general, the  single-particle Hamiltonian $H$
is defined by trigonometric functions, but near a fixed point    $H$ is   the  sum of a Dirac operator and a quadratic correction term.
Let us look at the adjoint of $\tilde{H}$,
\begin{equation*}
   \tilde{H}^* = \begin{pmatrix}
        0 & H^* \\
        \Theta H^* \Theta^*  & 0
      \end{pmatrix}
      =  \begin{pmatrix}
        0 & H \\
        \Theta H \Theta^*  & 0
      \end{pmatrix}
\end{equation*}
In the above, if $H$ is approximated by a Dirac operator $D$, which has the natural property $\Theta D \Theta^* = - D$,
then $\tilde{H}$  becomes a skew-adjoint  operator, i.e., $\tilde{H}^*  = - \tilde{H}$.
After adding a quadratic correction term to  $H$, the effective Hamiltonian
$\tilde{H}$ still belongs to the same homotopy class as that of $\begin{pmatrix}
                                                                  0 & -D \\
                                                                  D & 0
                                                                 \end{pmatrix}$,
since $H$ can be viewed as a continuous deformation of the Dirac operator $D$.
In other words, the local form of $\tilde{H}$ near a fixed point is a continuous deformation of a skew-adjoint operator by a small quadratic term.
The index problem of $\tilde{H}$ is determined by localizations around the fixed points.
Indeed, the physical property of a topological insulator is determined  by the localized form of $\tilde{H}$ around the fixed points, i.e., Majorana zero modes.
\end{proof}

\begin{rmk}
 An alternative way to construct the effective Hamiltonian is a  bottom-up approach. Since an edge state describes an electronic state, it can be modeled by a Dirac operator.
 When restricted to a fixed point, it is still a Dirac operator. In fact, a Majorana zero mode lives in a small neighborhood of a fixed point, so we have to extend the Dirac operator
 in one extra dimension and call it $D$, and then get the mirror image $-D$ under time reversal symmetry. This strategy of constructing a new Dirac operator is commonly used in the literature,
 for details the reader is referred to \cite{FSFF12, L88}.
 Once we have the skew-adjoint operator  $\begin{pmatrix}
                                                0 & -D \\
                                                D & 0
                                          \end{pmatrix}$ around a fixed point, it can be extended trivially 
 between two fixed points and finally obtain the effective Hamiltonian over a circle (more generally on a torus or sphere).  
 
\end{rmk}

\begin{thm}
   The topological $\mathbb{Z}_2$ invariant $\nu$ is the mod 2 analytical index of the effective Hamiltonian $\tilde{H}$,
\begin{equation}
 \nu = ind_a(\tilde{H})
\end{equation}
\end{thm}
\begin{proof}
  Since $\tilde{H}$ is a continuous deformation of a  skew-adjoint operator and the analytical index map is homotopy invariant, the mod 2 index of $\tilde{H}$ is well defined,
  $$
 ind_a(\tilde{H}) = \dim \ker \tilde{H}  \quad \text{mod 2}
  $$
  Suppose $\Psi = (\psi, \Theta \psi) \in \ker \tilde{H}$, i.e.,
  $$
   \begin{pmatrix}
      0 & \Theta H \Theta^* \\
      H & 0
   \end{pmatrix}
   \begin{pmatrix}
     \psi \\
     \Theta \psi
   \end{pmatrix} =
   \begin{pmatrix}
      0 \\
      0
   \end{pmatrix}
  $$
  then $H(x)\psi(x) = 0$ and $H(\tau(x)) \Theta \psi(x) = 0$. It only happens when $\Psi$ is a Majorana zero mode around a fixed point $x \in X^\tau$.
  Thus the parity of Majorana zero modes, i.e., the topological $\mathbb{Z}_2$ invariant, can be interpreted as the mod 2 index $ind_a(\tilde{H})$.  

\end{proof}

\begin{rmk}
Notice that we actually take the quaternionic dimension in the above analytical index since a zero mode of $H$ is a complex state,
$$
ind_a(\tilde{H}) = \dim_\mathbb{H} \ker \tilde{H}  \quad \text{mod 2}
$$
Because of the Kramers degeneracy, a Majorana zero mode consists of two complex chiral zero modes of $H$, that is,
$\dim_\mathbb{H} \ker \tilde{H} = \dim_\mathbb{C} \ker H$. Furthermore, zero modes of $H$ can only be found around a fixed point, where $H$ is effectively approximated by a Dirac operator $D$, so
$\dim_\mathbb{C} \ker H = \dim_\mathbb{C} \ker D$. Putting it together, we have
\begin{equation}
  ind_a(\tilde{H}) = dim_\mathbb{C} \ker D \quad \text{mod 2}
\end{equation}

\end{rmk}

The classifying spaces of KR-groups are constructed based on  subspaces of (skew-adjoint) Fredholm operators in \cite{AS69}, and they
 are also linked to different $\mathbb{Z}_2$ symmetries of general topological insulators in \cite{GS15}.
In the next paragraph, we identify the above analytical index as an element in $KO^{-2}(pt)$.

Let $\mathcal{F} (\mathscr{H}, J)$ denote the space of Fredholm operators  on a Real Hilbert space
$(\mathscr{H}, J)$, where $\mathscr{H}$ is a complex Hilbert space and $J$ is a real structure such that $J^2 =  1$. In addition, let
$\hat{\mathcal{F}} (\mathscr{H}, J)$ denote the subspace of skew-adjoint Fredholm operators,
 so $\hat{\mathcal{F}} (\mathscr{H}, J)$ is a classifying space of $KR^{-1}$,
 i.e., $KR^{-1}(X) = [X, \hat{\mathcal{F}} (\mathscr{H}, J)]$ for a compact Real space $(X, \tau)$.
In our case, the effective Hamiltonian $\tilde{H}$ is a continuous deformation of a skew-adjoint Fredholm operator acting on the
Real  Hilbert space $(L^2(X, {\mathcal{H}}), \mathcal{J})$,
so the analytical index belongs to $KO^{-2}(pt)$,
\begin{equation}
   ind_a:  \hat{\mathcal{F}} (L^2(X, {\mathcal{H}}), \mathcal{J}) \rightarrow  KO^{-2}(pt) = \mathbb{Z}_2
\end{equation}
Note that if  $\mathscr{H}$ is a complex Hilbert space, then the analytical index  lives in $KO^{-1}(pt)$.
However, in our case  $L^2(X, {\mathcal{H}})$ is viewed as a  Hilbert
space over quaternions $\mathbb{H}$, since for a Kramers pair $(\phi, \Theta \phi)$ and any point $x \in X$,  $(\phi(x), \Theta \phi(x)) \in \mathbb{H} \cong \mathbb{C} \oplus \Theta \mathbb{C}$, so $ind_a(\tilde{H}) \in KO^{-2}(pt)$.

In the following, we look at the mod 2 spectral flow related to the mod 2 analytical index.
Some relevant details on the spectral flow are given below,  for example see \cite{P96} for a general discussion of the spectral flow of self-adjoint operators.
Around a fixed point, a chiral state $\phi$ and its mirror partner $\Theta \phi$ may intersect with each other.
In a localized Majorana zero mode $\Phi_0 = (\phi_0, \Theta \phi_0)$,  a chiral zero mode $\phi_0$ or  $\Theta \phi_0$
will change the sign of its eigenvalue after going across the zero energy level.
Fix a chiral zero mode $\phi_0$ or $\Theta \phi_0$ in a Majorana zero mode,
if its sign  changes
from negative to positive when passing through a fixed point, then the spectral flow of the chosen chiral zero mode will increase by $1$.
On the other hand, if its sign changes in the negative direction,
 the spectral flow of the chiral zero mode will decrease by 1. So around a fixed point, the spectral flow of a chiral zero mode could change by $-1, 0$ or $1$.

 However, there is no a priori way  to tell $\phi$ and $\Theta \phi$ apart  since they are mirror partners of each other under the time reversal symmetry.
Time reversal symmetry does not determine whether a chiral state $\phi$ or $\Theta \phi$ is left-moving or right-moving, but only reverses the chirality,
that is, interchanges between  left-moving and right-moving.
As a consequence, if we pick  a chiral zero mode $\phi_0$ or $\Theta \phi_0$ around a fixed point,  the increase or decrease of
 the spectral flow of the chosen chiral zero mode at that fixed point by $1$ should  be equivalent, i.e., $+1 \equiv -1$ mod 2.
 So the existence of a Majorana zero mode can be counted by either adding or subtracting $1$ to the spectral flow of a chiral zero mode. In the end,
 we only count the parity of Majorana zero modes, so adding or subtracting $1$ to the mod 2 spectral flow are equivalent.

With the above physical picture in mind, let us consider the  spectral flow of the single-particle Hamiltonian $H(x)$ for $x \in X$, which is a family of self-adjoint Fredholm operators.
Heuristically, the spectral flow of a one parameter family of self-adjoint Fredholm operators is just the net number of eigenvalues (counting multiplicities) which pass
through zero in the positive direction from the start of the path to its end. More precisely, let $B: [0,1] \rightarrow \mathscr{F}^{sa}_*$ be a continuous path of self-adjoint Fredholm operators, if one defines
eigenprojections by functional calculus $E_a(t) : = \chi_{[-a, a]}(B_t)$, then the spectral flow of $\{B_t\}$, denoted by $sf(B)$, can be defined  as the dimension of the nonnegative eigenspace
at the end of this path minus the dimension of the nonnegative eigenspace at the beginning, for details see \cite{P96}. The space of self-adjoint Fredholm operators $\mathscr{F}^{sa}$ has three components, the 1st component
$\mathscr{F}^{sa}_+$ has essential spectrum $\{+1\}$, the 2nd component $\mathscr{F}^{sa}_-$ has essential spectrum $\{-1\}$, and the 3rd component is the complement
$\mathscr{F}^{sa}_* = \mathscr{F}^{sa} \setminus (\mathscr{F}^{sa}_+ \sqcup \mathscr{F}^{sa}_-)$ \cite{AS69}.


\begin{thm}
   The analytical  index $ind_a(\tilde{H})$ can be computed by the spectral flow of the single-particle Hamiltonian $H$ modulo 2.
\end{thm}
\begin{proof}
  Around a fixed point, if we can find a Majorana zero mode, then we use the spectral flow of a chiral zero mode to count that Majorana zero mode.
  Running through all the fixed points, we count all Majorana zero modes by adding those spectral flows together,
  which is an integer between $-k_0$ and $k_0$ where $k_0$ is the number of fixed points. In the end, the parity of Majorana zero modes is
  the collective result of those spectral flows of chiral zero modes modulo $2$.

  The Hamiltonian $H$ is a self-adjoint Fredholm operator parametrized by $x \in X$. Each eigenstate with zero energy gives a chiral zero mode at a fixed point.
  For a Majorana zero mode around a fixed point $x_k$, one takes a path  $H: I_k= [0, 1] \rightarrow \mathscr{F}^{sa}_*$ such that the fixed point $x_k$ is located at $t= \frac{1}{2}$,
  so the spectral flow of a chiral zero mode is the same as the spectral flow of $\{ H_t \}$ for $t \in I_k$.
  Note that the spectral flow of  $\{ H_t \}$ is independent of the choice of intervals $I_k$ as long as such interval passes through $x_k$. One can find sketches of such spectral flows for example in \cite{DS16}.
  For each fixed point with a Majorana zero mode, one considers a path $\{ H_t \}$ for $t \in I_k$ and its local spectral flow, where $k$ labels the fixed points with Majorana zero modes.
  So the spectral flow of the Hamiltonian $H(x)$ for $x \in X$ is the sum of all local spectral flows of $\{ H(t), t \in I_k \}$ for all $k$,
  which  counts all the sign changes of those local chiral zero modes running over all the fixed points. Therefore,
  the spectral flow of the Hamiltonian $H(x)$ modulo 2 computes the parity of Majorana zero modes, i.e., the mod 2 analytical index $ind_a(\tilde{H})$.

 \end{proof}


\begin{rmk}
   In the symplectic setting, if one models a chiral state by some Lagrangian submanifold, then the intersection number between this Lagrangian and the zero energy level is given by the Maslov index,
   which is used to define an edge $\mathbb{Z}_2$ index in \cite{ASV13}.
   The Maslov index can be geometrically realized by the spectral flow of a family of Dirac operators \cite{KLW15}, so that the edge $\mathbb{Z}_2$ index can be computed by a mod 2 spectral flow. Our approach
   has essentially the same idea using spectral flow, but it is not necessary to set up the model in symplectic topology.
\end{rmk}

\begin{rmk}
   Due to the construction of the effective Hamiltonian $\tilde{H}$, 
   the spectral flow of $\tilde{H}$ (counting  Majorana zero modes) is reduced to the spectral flow of the self-adjoint Fredholm operator $H$ (counting chiral zero modes). 
   For a path of skew-adjoint Fredholm operators, a mod 2 spectral flow is constructed in \cite{CPS16}.
\end{rmk}

\subsection{KQ-cycle} \label{sec:KQcyc}

Similar to a KR-cycle modeling a Majorana spinor,  we define a KQ-cycle to model a Majorana quasi-particle based on the effective Hamiltonian $\tilde{H}$ acting on a Kramers pair. 
There exists a natural grading operator
$\gamma = \begin{pmatrix}
          1 & 0 \\
          0 & -1
         \end{pmatrix}$
with $\gamma = \gamma^*$ and $\gamma^2 = 1$,  so $\gamma$ can be used to separate the chiral states in a Kramers pair $\Phi = (\phi, \Theta \phi)$,
$$
 \frac{1 + \gamma}{2} \Phi = \begin{pmatrix}
                               1 & 0 \\
                               0 & 0
                             \end{pmatrix} \begin{pmatrix}
                                             \phi \\
                                             \Theta \phi
                                            \end{pmatrix} = \begin{pmatrix}
                                                               \phi \\
                                                                0
                                                             \end{pmatrix}, \quad
 \frac{1 - \gamma}{2} \Phi = \begin{pmatrix}
                               0 & 0 \\
                               0 & 1
                             \end{pmatrix} \begin{pmatrix}
                                             \phi \\
                                             \Theta \phi
                                            \end{pmatrix} = \begin{pmatrix}
                                                               0 \\
                                                                \Theta \phi
                                                             \end{pmatrix}
$$
The Hilbert space $L^2(X,  {\mathcal{H}})$ is naturally $\mathbb{Z}_2$-graded
and can be decomposed into two chiral components $L^2(X, \mathcal{L}) \oplus  L^2(X, \Theta \mathcal{L})$ provided $\mathcal{H} \cong \mathcal{L} \oplus \Theta \mathcal{L}$.

\begin{defn}\label{KQcycle}
 The KQ-cycle of a Kramers pair is defined as the quintuple
\begin{equation}\label{KQ6cyc}
   (C^\infty(X), L^2(X, \mathcal{H})  , \tilde{H}, \mathcal{J}, \gamma )
\end{equation}
where
$$
\tilde{H} = \begin{pmatrix}
               0 & \Theta H \Theta^* \\
                H& 0
               \end{pmatrix}, \quad 
                                   \mathcal{J} = \begin{pmatrix}
                                                0 & \Theta^* \\
                                              \Theta & 0
                                                  \end{pmatrix} 
$$
\end{defn}

   Similar to the representation of spinors when $2k = 2$ mod 8, 
   the $KQ$-cycle models two inequivalent complex spinors $(\phi, \Theta \phi) \in L^2(X, \mathcal{L}) \oplus  L^2(X, \Theta \mathcal{L})$,
   and a Kramers pair can be viewed as a quaternionic pinor $\Phi \in L^2(X,  {\mathcal{H}})$. 
      The localization of the $KQ$-cycle (i.e., a localized Majorana zero mode) will be viewed as a generalized $KR_{10}$-cycle  with KO-dimension $j = 2$ mod 8.

 \begin{rmk}As a convention in condensed matter physics, the notation $H(\mathbf{k})$ for $\mathbf{k} \in K$     is commonly used for a Hamiltonian, where
 $K$ is called the $k$-space, i.e., the momentum space, for examples of $H(\mathbf{k})$ see \cite{S13}. In a topological insulator, near a fixed point $H(\mathbf{k}) = \mathbf{k} \cdot \vec{\sigma} + O(\mathbf{k}^2) $
 for the $\mathbf{k}$-vector $\mathbf{k} = (k_1, k_2)$  and Pauli matrices $\vec{\sigma} = (\sigma_1, \sigma_2)$,
 so  the single-particle Hamiltonian ${H}$ can be approximated by a Dirac operator $D$ by localization.
 For the spectral triple and noncommutative calculus in the momentum representation, we follow the canonical conventions on the noncommutative Brillouin torus, for example see \cite{PLB13}.
 Hence the above $KQ$-cycle is different from a classical $KR$-cycle in that (1) it is constructed to model quasi-particles (e.g. a Kramers pair) instead of real particles (e.g. Majorana spinors); 
 (2) it  goes beyond the scope of Dirac geometry, where only Dirac operators (or spinors) are involved;
 (3) the single-particle Hamiltonian $H$ is represented over the momentum space,  by localization, $H$ can be
 approximated  by a Dirac operator $D$ around a fixed point;
 (4) the effective Hamiltonian $\tilde{H}$ is not a self-adjoint operator,    but can be  approximated by a skew-adjoint operator around a fixed point. 
 \end{rmk}

\begin{prop}
   The localization of the  KQ-cycle  $(C^\infty(X), L^2(X, \mathcal{H})  , \tilde{H}, \mathcal{J}, \gamma )$   describes the geometry of a localized Majorana zero mode, 
   which is a coupled product of two $KR_5$-cycles. 
\end{prop}

\begin{proof} 
Around a fixed point, the single-particle Hamiltonian $H$ is approximated by a self-adjoint Dirac operator $D$, so the effective Hamiltonian $\tilde{H}$  can be approximated by a skew-adjoint operator $ \tilde{H} \sim \begin{pmatrix}
                                                                   0 & -D \\
                                                                  D & 0
                                                                 \end{pmatrix}$.  
Define a KR-cycle by  the quadruple, called the associated KR-cycle,   
\begin{equation}
  (C^\infty(X), L^2(X, \mathcal{L}) \oplus  L^2(X, \Theta \mathcal{L}), D, \Theta)
\end{equation}
which models a chiral state in a localized Majorana zero mode.
Notice that the time reversal operator $\Theta: L^2(X, \mathcal{L}) \rightarrow L^2(X, \Theta \mathcal{L})$ changes the chirality  of a chiral state $\phi \in L^2(X, \mathcal{L})$.
With the  real structure $\Theta$,  this  is a $KR_5$-cycle satisfying
\begin{equation} \label{KR5cycle}
   \Theta D = - D \Theta, \quad \Theta^2 = -1
\end{equation}
If one uses the other chiral state $\Theta \phi \in  L^2(X, \Theta \mathcal{L})  $ to define a KR-cycle, one obtains a second $KR_5$-cycle
\begin{equation}
  (C^\infty(X),    L^2(X, \Theta \mathcal{L}) \oplus L^2(X, \mathcal{L}), -D, -\Theta)
\end{equation}
As a generalization of the approximating operator $\begin{pmatrix} 
                                            0 & -D \\
                                             D & 0
                                           \end{pmatrix}$, 
 at the level of KR-cycles, 
 the localization of the KQ-cycle is a coupled product  of two $KR_5$-cycles,
 which is viewed as a generalized $KR_{10}$-cycle with KO-dimension 2 (= 10 mod 8).

\end{proof}

\section{Topological index}\label{sec:Tind}
The Atiyah--Singer index theorem teaches us that the
analytical index can be computed by the topological index.
In this section, we will compute the mod 2 analytical index
by a mod 2 topological index.
The key observation is that the parity anomaly of the topological $\mathbb{Z}_2$ invariant can be translated into a gauge anomaly, 
and the local formula is basically given by the odd topological index of a specific gauge representing time reversal symmetry.
In this paper, we only give examples in 2d and 3d, which are the cases of interest for condensed matter physics.

\subsection{Topological index map} \label{sec: Tindmap}
Given a skew-adjoint elliptic operator $P$ with the  symbol class $[\sigma(P)] \in KR^{-2}(T^*X)$,
  the topological index of $[\sigma(P)]$ was constructed by Atiyah  \cite{A76},
\begin{equation}
   ind_t : KR^{-2}(T^*X) \rightarrow KO^{-2}(pt) = \mathbb{Z}_2
\end{equation}
where $\pi: T^*X \rightarrow X$ is the cotangent bundle over $X$.
For a $d$-dimensional involutive space $(X, \tau)$, the Thom isomorphism in $KR$-theory is given by
\begin{equation}
   KR^{-j}(X) \cong KR^{d-j}(T^*X)
\end{equation}
Combining these two maps gives a map from KR-theory (or KQ-theory) to $KO^{-2}(pt)$, and we still call it the topological index map.

\begin{examp}
   When $X = \mathbb{T}^2$, the topological index map is a map from
   $KQ(\mathbb{T}^2)$ to $\mathbb{Z}_2$ since the Thom isomorphism identifies $ KQ(\mathbb{T}^2) $ with $KR^{-2}(T^*\mathbb{T}^2)$,
   $$
    ind_t: KQ(\mathbb{T}^2) = KR^{-4}(\mathbb{T}^2)  \cong KR^{-2}(T^*\mathbb{T}^2) \rightarrow KO^{-2}(pt) 
   $$

\end{examp}

\begin{examp}
\label{indexex}

  When $X = \mathbb{T}^3$, the topological index map is a map from $KQ^{-1}(\mathbb{T}^3) $ to $\mathbb{Z}_2$,
   $$
    ind_t: KQ^{-1}(\mathbb{T}^3) = KR^{-5}(\mathbb{T}^3)  \cong KR^{-2}(T^*\mathbb{T}^3)  \rightarrow KO^{-2}(pt) 
   $$

\end{examp}

 From the KQ-groups of $\mathbb{T}^3$, $\widetilde{KQ}(\mathbb{T}^3) (=  3KO^{-2}(pt) \oplus KO^{-1}(pt))$
   and $ \widetilde{KQ}^{-1}(\mathbb{T}^3) ( = 3KO^{-4}(pt)\oplus KO^{-2}(pt))$ both contain  $\mathbb{Z}_2$ components.
   But the topological index map identifies $KO^{-2}(pt)$ as the right place where the topological $\mathbb{Z}_2$ invariant really lives in.
   This fact matches perfectly with the results on the effective Hamiltonian $\tilde{H}$ and its analytical index. 
 
In general, the topological index can be computed based on the Chern character, which is a map from complex K-theory to
de-Rham cohomology by  Chern--Weil theory.
In the following subsections, we will see concrete realizations of the mod 2 topological index in 3d and 2d.

\subsection{3d case}\label{3dTopInd}
The topological $\mathbb{Z}_2$ invariant is a parity anomaly in quantum field theory, which is a global anomaly and hard to compute.
In three dimensions, based on the Chern--Simons theory, this parity anomaly is translated into
a gauge anomaly, which is a locally computable gauge problem. This idea has been explained in detail for example in \cite{KLW15}, which has the origin from the $SU(2)$ gauge anomaly by Witten \cite{W82}.

For the 3d momentum spaces $X = \mathbb{T}^3$ or $\mathbb{S}^3$, we assume the rank of the Hilbert bundle $\pi: \mathcal{H} \rightarrow X$ is 2, i.e.,
$rank(\mathcal{H}) = 2$. Recall that the transition function $w$ in  Eq.\eqref{TransFunc} defines a map
$w: X \rightarrow U(2)$, so the structure (or gauge) group is $U(2)$. If, in addition, the Majorana states $\Phi$ are assumed to be normalized such that $\langle \Phi, \Phi \rangle = 1$,
then the gauge group  $U(2)$ is reduced to $SU(2) \simeq Sp(1)$.

The (anti-)involutions of the Hilbert bundle $({\mathcal{H}}, \Theta) \rightarrow (X, \tau)$ induces
 an involution on the structure group. Define $ \sigma: U(2) \rightarrow U(2)$ by $\sigma(g) \mapsto -g^T$ so that $\sigma^2 = 1$.
The natural compatibility condition is given by   $w \circ \tau = \sigma \circ w$, that is,  $w$ defines an equivariant map,
$$
w: (X, \tau) \rightarrow (U(2), \sigma)
$$
Furthermore, $w$ gives rise to  the generator $[w]$ of  $\mathbb{Z}_2 \in \widetilde{KQ}^{-1}(X)$ for $X = \mathbb{S}^3$ or $\mathbb{T}^3$, 
which has an intimate relation to spectral flow \cite{G93, L88}.
We study the geometry and topology of time reversal symmetry, which is represented by the transition function $w$. 
The argument about spectral flow guarantees the K-theoretic class $[w]$ falls into $KQ^{-1}(X)$ (a shift by $-1$ from $KQ(X)$), 
so the topological $\mathbb{Z}_2$ invariant has an interpretation as an odd topological index of $w$.   
In physical terms, we consider the gauge theory of a topological insulator, and $w$ is the gauge transformation (induced by the time reversal symmetry)
characterizing the band structure.

The odd Chern character of a differentiable map $g: X \rightarrow U(n)$ is defined by
\begin{equation*}
   Ch(g) := \sum_{k = 0}^{(\dim X - 1)/2} Ch_{2k+1}(g) = \sum_{k=0}^{(\dim X - 1)/2} (-1)^k \frac{k!}{(2k+1)!} tr[(g^{-1}dg)^{2k+1}]
\end{equation*}
which is a closed form of odd degree  \cite{G93}.
The topological index in odd dimensions can be computed by the odd index theorem \cite{BD82}, when the Dirac or A-roof genus
$\hat {A}(X) = 1 $, 
\begin{equation}\label{OddIndex}
   ind_t(g)  =      \frac{1}{4\pi^2 }   \int_X Ch_3(g) = -\frac{1}{24\pi^2 }   \int_X tr(g^{-1}dg)^{3}
\end{equation} 
for the three-dimensional case.
This formula is sometimes called the winding number (or degree) of $g$, which  can be used to compute the spectral flow $ sf(D, g^{-1}Dg)$ \cite{G93, KLW15}.

\begin{prop} \label{prop2}
  The odd topological index of the transition function $w$, i.e., $ind_t(w)$, is naturally $\mathbb{Z}_2$-valued.
\end{prop}
 \begin{proof} 
 Up to a normalization constant, the odd index (or the winding number) of $w$  is basically given by
 $$
  \int_X tr(w^{-1}dw)^{3} 
 $$
 Now we apply the time reversal transformation $\tau$ to $X$, and change the local coordinates from $x$ to $\tau(x)$.
 $$
 \int_X tr [w^{-1}(\tau(x))dw(\tau(x)) ]^{3}
 $$
 Using the compatibility condition between  $w$  and  $\tau$, i.e.,
  $$
  w \circ \tau = \sigma \circ w = \sigma(w) =  -w^{T} = - \bar{w}^{-1}
  $$ the above equals
    \begin{equation*}
        \int_X tr(\bar{w}d \bar{w}^{-1})^{3}   =  \overline{ \int_X tr(wdw^{-1})^{3} } = \int_X tr(wdw^{-1})^{3} 
    \end{equation*}
     which gives  the  winding number of $w^{-1}$. 
     It is well-known that $w$ and $w^{-1}$ have opposite winding numbers, i.e.,
     $$
     ind_t(w^{-1} ) = - ind_t(w)
     $$
     As a global invariant, the winding number does not depend on the choice of local coordinates, so for the involutive space $(X, \tau)$, 
  the topological index $ind_t(w)$ should be identified with $-ind_t(w)$ from the above computation. 
  
    In other words, the odd topological index of $w$  is a mod 2 degree,
    $$
    ind_t(w) = deg_2(w)
    $$
    since $ind_t(w)$ must be $\mathbb{Z}_2$-valued due to the time reversal symmetry.
 \end{proof}

 \begin{rmk}

    In string theory, $ind_t(g)$ is also called the  Wess--Zumino--Witten (WZW) topological term.
    In the seminal work \cite{W82}, Witten pointed out that as a $SU(2)$ anomaly   $ind_t(g)$ is actually $\mathbb{Z}_2$-valued and related to the mod 2 index.
    So the above result is already known in the physics literature, and we just gave a new proof in the presence of the time reversal symmetry.
   
 \end{rmk}

The odd index theorem gives a local formula to compute the topological $\mathbb{Z}_2$ invariant in three dimensions.
Similar to the Atiyah--Singer index theorem, we obtain the mod 2 index theorem for 3d topological insulators, i.e.,
the analytical index equals the topological index.

\begin{thm}
  The topological $\mathbb{Z}_2$ invariant for 3d topological insulators can be understood as a mod 2 index theorem,
\end{thm}
\begin{equation}
  ind_a(\tilde{H}) = ind_t(w)
\end{equation}
\begin{proof}
   First of all, the analytical index of $\tilde{H}$ counts the parity of Majorana zero modes, which can be computed by the mod 2 spectral flow of the self-adjoint Fredholm operator $H$.
   On the other hand, the odd index formula $ind_t(w)$ computes the spectral flow $sf(D, w^{-1}Dw)$ \cite{G93}, 
   where $D$ is the Dirac operator used to approximate  $H$ by localization, 
   and $ind_t(w)$  is naturally $\mathbb{Z}_2$-valued. So we can prove the mod 2 analytical index and the mod 2 topological index are the same
   based on different interpretations of the same mod 2 spectral flow.
\end{proof}

The WZW topological term is an action functional  of the relevant gauge theory.
Kane and Mele considered the effective fermionic field theory of a topological insulator,
and derived a Pfaffian formalism of the topological $\mathbb{Z}_2$ invariant, which is  called the Kane--Mele invariant.
We now show the equivalence between these two formalisms of the topological $\mathbb{Z}_2$ invariant. 

\begin{defn}(\cite{KM05}) Assume the set of fixed points $X^\tau$ is finite,
the Kane--Mele invariant of a topological insulator is defined by
   \begin{equation} \label{KMinv}
   \nu  = \prod_{x \in X^\tau} \frac{pf [w(x)]}{\sqrt{\det [w(x)]}}
  \end{equation}
\end{defn}
Recall that at any fixed point $x\in X^\tau$,  the transition function $w$ from Eq.\eqref{TransFunc} 
is a skew-symmetric matrix, i.e., $ w^T(x) =- w ({x})$, where $T$ is the transpose of a matrix. So it makes sense to
take the Pfaffian of $w$ at any fixed point $x \in X^\tau$, denoted by $pf[w(x)]$.
The relation between the Pfaffian and determinant function is $pf^2(M) = \det(M)$ for a skew-symmetric matrix $M$.
Compared with the square root of the determinant of $w$ at $x \in X^\tau$, i.e., $\sqrt{\det[w(x)]}$,
the Kane--Mele invariant is defined as the product of the signs of Pfaffians over the fixed points.

Let us start with the identity $\ln \det (A) = tr \ln (A)$ for a square matrix $A$,
when $A$ is skew-symmetric, it becomes $2 \ln pf (A) =  tr \ln (A)$.
The transition function  $w$ is  a matrix parametrized by a 3d  momentum space  $w: (X, \tau) \rightarrow (U(2), \sigma)$,
and we consider the variation of $\ln \det (w) $
with respect to the local coordinates of $X$.

\begin{lemma}
   The integral form of the variation  of $\ln \det (w) $ equals  the odd topological index of $w$ up to a constant,
\begin{equation} \label{DetEqWn}
   \int_X d \ln \det (w) = -\frac{1}{4} \int_X tr(w^{-1}dw)^3
\end{equation}
\end{lemma}
\begin{proof}
  Suppose an invertible and differentiable matrix $A(x, y, z)$ depends on  variables $x, y, z$, we consider the third partial derivative of $\ln \det (A) $ with respect to $x, y$ and $z$.
  The first partial derivative is well-known,
  $$
  \frac{\partial \ln \det A}{ \partial x} =\frac{\partial \, tr \ln A}{ \partial x} = tr (A^{-1} \frac{\partial A}{\partial x})
  $$
  The second partial derivative is
  $$
  \begin{array}{ll}
   \frac{\partial^2  \ln \det A} { \partial x \partial y}
            & = tr (A^{-1} \frac{\partial ^2 A}{\partial x \partial y}) + tr(\frac{\partial A^{-1}}{\partial y} \frac{\partial A}{\partial x} )\\
           & = tr (A^{-1} \frac{\partial ^2 A}{\partial x \partial y}) - tr(A^{-1} \frac{\partial A}{\partial y} A^{-1}  \frac{\partial A}{\partial x} )
  \end{array}
  $$
  since $\frac{\partial A^{-1}} {\partial y}  A + A^{-1} \frac{\partial A} {\partial y} = 0$ based on $A^{-1} A = I$, i.e., $\frac{\partial A^{-1}} {\partial y}  = -  A^{-1} \frac{\partial A} {\partial y}  A^{-1}$.
   The third partial derivative is
  $$
  \begin{array}{lll}
   \frac{\partial^3  \ln \det A} { \partial x \partial y \partial z}
           & = & tr (A^{-1} \frac{\partial ^3 A}{\partial x \partial y \partial z}) + tr (\frac{ \partial A^{-1}}{\partial z} \frac{\partial ^2 A}{\partial x \partial y})  \\
           & & - tr(\frac{ \partial A^{-1}}{\partial z} \frac{\partial A}{\partial y} A^{-1}  \frac{\partial A}{\partial x} )  - tr(A^{-1} \frac{\partial^2 A}{\partial y \partial z} A^{-1}  \frac{\partial A}{\partial x} ) \\
           & & - tr(A^{-1} \frac{\partial A}{\partial y} \frac{ \partial A^{-1}}{\partial z}  \frac{\partial A}{\partial x} ) - tr(A^{-1} \frac{\partial A}{\partial y} A^{-1}  \frac{\partial^2 A}{\partial x \partial z} ) \\
           & = & tr(\frac{ \partial \ln A}{\partial x} \frac{\partial \ln A}{\partial y}   \frac{\partial \ln A}{\partial z} ) + tr(\frac{ \partial \ln A}{\partial x} \frac{\partial \ln A}{\partial z}   \frac{\partial \ln A}{\partial y} ) + tr (A^{-1} \frac{\partial ^3 A}{\partial x \partial y \partial z}) \\
           & &  +  tr ( \frac{ \partial A^{-1}}{\partial x} \frac{\partial ^2 A}{\partial y \partial z} ) + tr ( \frac{\partial ^2 A}{\partial x \partial z} \frac{ \partial A^{-1}}{\partial y}) + tr ( \frac{\partial ^2 A}{\partial x \partial y} \frac{ \partial A^{-1}}{\partial z})
  \end{array}
  $$
  so the differential of $\ln \det A$ is
   $$
   \begin{array}{lll}
   d \ln \det A &= & [tr(\frac{ \partial \ln A}{\partial x} \frac{\partial \ln A}{\partial y}   \frac{\partial \ln A}{\partial z} ) + tr(\frac{ \partial \ln A}{\partial x} \frac{\partial \ln A}{\partial z}   \frac{\partial \ln A}{\partial y} )
                 + tr (A^{-1} \frac{\partial ^3 A}{\partial x \partial y \partial z}) \\
                &  & +  tr ( \frac{ \partial A^{-1}}{\partial x} \frac{\partial ^2 A}{\partial y \partial z} ) + tr ( \frac{\partial ^2 A}{\partial x \partial z} \frac{ \partial A^{-1}}{\partial y}) + tr ( \frac{\partial ^2 A}{\partial x \partial y} \frac{ \partial A^{-1}}{\partial z})] dxdydz \\
                &= & [tr(\frac{ \partial \ln A}{\partial x} \frac{\partial \ln A}{\partial y}   \frac{\partial \ln A}{\partial z} )  + tr (A^{-1} \frac{\partial ^3 A}{\partial x \partial y \partial z}) \\
                &  & +  tr ( \frac{ \partial A^{-1}}{\partial x} \frac{\partial ^2 A}{\partial y \partial z} ) +                tr ( \frac{\partial ^2 A}{\partial x \partial y} \frac{ \partial A^{-1}}{\partial z})] dxdydz \\
                &  & - [tr(\frac{ \partial \ln A}{\partial x} \frac{\partial \ln A}{\partial z}   \frac{\partial \ln A}{\partial y} )  + tr ( \frac{\partial ^2 A}{\partial x \partial z} \frac{ \partial A^{-1}}{\partial y})] dxdz dy
  \end{array}
   $$

   If we change the order of the variables $x, y$ and $z$, we similarly get
   $$
     \begin{array}{lll}
   d \ln \det A  = & [tr(\frac{ \partial \ln A}{\partial y} \frac{\partial \ln A}{\partial x}   \frac{\partial \ln A}{\partial z} )
                   + tr (A^{-1} \frac{\partial ^3 A}{\partial y \partial x \partial z}) \\
                   & +  tr ( \frac{ \partial A^{-1}}{\partial y} \frac{\partial ^2 A}{\partial x \partial z} )
                     +    tr ( \frac{\partial ^2 A}{\partial y \partial x} \frac{ \partial A^{-1}}{\partial z})] dydxdz \\
                   &- [tr(\frac{ \partial \ln A}{\partial y} \frac{\partial \ln A}{\partial z}   \frac{\partial \ln A}{\partial x} )  + tr ( \frac{\partial ^2 A}{\partial y \partial z} \frac{ \partial A^{-1}}{\partial x})] dydz dx
  \end{array}
   $$
   and
     $$
   \begin{array}{lll}
   d \ln \det A  = & [tr(\frac{ \partial \ln A}{\partial z} \frac{\partial \ln A}{\partial y}   \frac{\partial \ln A}{\partial x} )
                     + tr (A^{-1} \frac{\partial ^3 A}{\partial z \partial y \partial x}) \\
                  & +  tr ( \frac{ \partial A^{-1}}{\partial z} \frac{\partial ^2 A}{\partial y \partial x} )
                    +  tr ( \frac{\partial ^2 A}{\partial z \partial y} \frac{ \partial A^{-1}}{\partial x})] dzdydx \\
                  & - [tr(\frac{ \partial \ln A}{\partial z} \frac{\partial \ln A}{\partial x}   \frac{\partial \ln A}{\partial y} )
                  + tr ( \frac{\partial ^2 A}{\partial z \partial x} \frac{ \partial A^{-1}}{\partial y})] dzdx dy
  \end{array}
   $$
   Combining them together,  we have
   $$
    \begin{array}{ll}
    - d \ln \det A = &tr(A^{-1}dA)^3 + 3 [tr (A^{-1} \frac{\partial ^3 A}{\partial x \partial y \partial z}) +    tr ( \frac{ \partial A^{-1}}{\partial x} \frac{\partial ^2 A}{\partial y \partial z} ) \\
                  & + tr ( \frac{\partial ^2 A}{\partial x \partial z} \frac{ \partial A^{-1}}{\partial y}) + tr ( \frac{\partial ^2 A}{\partial x \partial y} \frac{ \partial A^{-1}}{\partial z})] dxdydz \\
    \end{array}
   $$
   where $tr(A^{-1}dA)^3 = \sum_{(i,j,k) \in S^3} \epsilon_{ijk} tr(d_i\ln A \, d_j \ln A \,  d_k \ln A) dx_i dx_j dx_k$.

   Now we replace $A$ by $w$ and take the integral,
    $$
    \begin{array}{ll}
    - \int_X d \ln \det w = & \int_X tr(w^{-1}dw)^3 + 3 \int_X [tr (w^{-1} \frac{\partial ^3 w}{\partial x \partial y \partial z} )+    tr ( \frac{ \partial w^{-1}}{\partial x} \frac{\partial ^2 w}{\partial y \partial z} ) \\
                  & + tr ( \frac{\partial ^2 w}{\partial x \partial z} \frac{ \partial w^{-1}}{\partial y}) + tr ( \frac{\partial ^2 w}{\partial x \partial y} \frac{ \partial w^{-1}}{\partial z})] dxdydz \\
    \end{array}
   $$
   Apply the compatibility condition $\tau^*w = - w^{T} = - \bar{w}^{-1}$ in the last integral, 
   $$
   \int_X tr [ \frac{\partial ^2 \tau^*w}{\partial x \partial y} \frac{ \partial \tau^*(w^{-1})}{\partial z}] dxdydz =  \int_X tr ( \frac{\partial ^2 w^{-1}}{\partial x \partial y} \frac{ \partial w}{\partial z}) dxdydz
   $$
   the complex conjugation is canceled since only the fixed points contribute to the integral and its value is a real number. 
   After plugging back into the above, we get
   $$
    \begin{array}{lll}
    - \int_X d \ln \det w & = & \int_X tr(w^{-1}dw)^3 + 3 \int_X [tr (w^{-1} \frac{\partial ^3 w}{\partial x \partial y \partial z} )+    tr ( \frac{ \partial w^{-1}}{\partial x} \frac{\partial ^2 w}{\partial y \partial z} ) \\
                  & & + tr ( \frac{\partial ^2 w}{\partial x \partial z} \frac{ \partial w^{-1}}{\partial y}) + tr ( \frac{\partial ^2 w^{-1}}{\partial x \partial y} \frac{ \partial w}{\partial z})] dxdydz \\
                   & =  &  \int_X tr(w^{-1}dw)^3 +  3 \int_X \frac{\partial }{\partial x} [tr(w^{-1} \frac{\partial ^2 w}{ \partial y \partial z} ) + tr ( \frac{\partial  w^{-1}}{ \partial y} \frac{ \partial w}{\partial z})] dxdydz \\
                   & = &  \int_X tr(w^{-1}dw)^3 +  3 \int_X \frac{\partial^2 }{\partial x \partial y}   tr (   w^{-1} \frac{ \partial w}{\partial z}) dxdydz \\
                   & = &  \int_X tr(w^{-1}dw)^3 +  3 \int_X\, d \,tr \ln w \\
  \end{array}
   $$
   Using  $ \ln \det w = tr \ln w $ again, we  obtain
   $$
    -4 \int_X d \ln \det w = \int_X tr(w^{-1}dw)^3
   $$

\end{proof}

  It is better to adjust the normalization constant by hand on the right hand side in \eqref{DetEqWn} to make it $\mathbb{Z}$-valued, which is a common practice in physics,
   \begin{equation}\label{DetEqInd}
      \int_X d \ln \det w = - \frac{1}{24 \pi^2}\int_X tr(w^{-1}dw)^3 = ind_t(w)
   \end{equation}
Let us look into this formula, and explain why  the Chern--Simons invariant, i.e., the topological index $ind_t(w)$, and the Kane--Mele invariant are equivalent.
On the left hand side, the determinant function can be interpreted as a section of the determinant line bundle, which is closely related to index theory.
The integral of the variation of $\ln \det w$ turns out to be the jumps at the fixed points since they are the only isolated singularities.
In other words,  the result of the left integral is the alternating difference between the evaluations on the fixed points,
\begin{equation*}
   \int_X d \ln \det w =  \Delta_{x \in X^\tau} \ln \det [w(x)]  = \Delta_{x \in X^\tau} 2 \ln  pf [w(x)]
\end{equation*}
where the determinant function is changed by the Pfaffian function since $w$ is skew-symmetric at the fixed points.
Since the odd topological index of $w$ is $\mathbb{Z}_2$-valued, the alternating difference can be replaced by a summation,
\begin{equation*}
   \int_X d \ln \det w  = \sum_{x \in X^\tau} 2 \ln pf [w(x)]
\end{equation*}

On the right hand side in \eqref{DetEqInd},   the topological index $ind_t(w)$  computes
the topological $\mathbb{Z}_2$ invariant. The mod 2 topological index takes the value $0$ or $1$, so it
is an element in the additive group $\mathbb{Z}_2$, i.e., $ind_t(w) \in (\mathbb{Z}_2, +)$.

The identity \eqref{DetEqInd}  gives rise to the  relation
\begin{equation} \label{AddInv}
   \sum_{x \in X^\tau}  \ln pf [w(x)] = \frac{1}{2}ind_t(w)   
\end{equation}
If we exponentiate both sides, we obtain
\begin{equation}
   \prod_{x \in X^\tau} pf [w(x)]  = exp \{ { {2 \pi i} \frac{ind_t(w)}{2} } \}  = (-1)^{ind_t(w)}
\end{equation}
where we put in a factor $2 \pi i$  on the right hand side since it is the result of an effective field theory.
In other words, the product of Pfaffians over the fixed points gives an exponentiated version of the topological $\mathbb{Z}_2$ invariant, denoted by $\nu$,
\begin{equation} \label{MulInv}
  \nu = \prod_{x \in X^\tau} pf [w(x)]  = (-1)^{ind_t(w)}
\end{equation}
As a consequence, the exponentiated topological $\mathbb{Z}_2$ invariant  is an element in the multiplicative group $\mathbb{Z}_2$, i.e., $\nu \in (\mathbb{Z}_2, \times)$.

By properties of the transition function $w$, the Pfaffian function $pf[w(x)]$  takes the value $1$ or $-1$ at any fixed point $x \in X^\tau$, 
since it is possible to normalize the determinant $\det[w(x)]$ to be 1.
So $\nu$ does not change if we replace the Pfaffians by their signs, i.e.,
\begin{equation*}
  \nu = \prod_{x \in X^\tau} sgn (pf [w(x)])
\end{equation*}
In general,  the square root of the determinant function  is added
as a reference term  to determine the sign of a Pfaffian, and the Kane--Mele invariant was originally defined as
\begin{equation*}
   \nu = \prod_{x \in X^\tau} \frac{pf [w(x)]}{\sqrt{\det [w(x)]}}
\end{equation*}

\begin{thm}
   The topological $\mathbb{Z}_2$ invariant for 3d topological insulators can be computed by the topological index,
   i.e., $ind_t(w)$. The Kane--Mele invariant $\nu$ is the exponentiated topological $\mathbb{Z}_2$ invariant,
   i.e., $\nu = (-1)^{ind_t(w)}$. They are equivalent since the additive group $(\mathbb{Z}_2, +)$ is isomorphic
   to the multiplicative group $(\mathbb{Z}_2, \times)$ by the exponential map.
\end{thm}
 The Chern--Simons invariant  and the Kane--Mele invariant provide two different ways to compute the topological $\mathbb{Z}_2$ invariant in 3d,
 the former is an action functional  of a specific gauge transformation
   and the latter is obtained from an effective quantum field theory. The equivalence of the Chern--Simons invariant and the Kane--Mele invariant was also proved
   in \cite{FM13, WQZ10} from different perspectives.

   From a different point of view, the equivalence relation \eqref{MulInv} between the topological index $ind_t(w)$ and the Kane--Mele invariant $\nu$ is viewed as the original bulk-boundary correspondence
   on the level of $\mathbb{Z}_2$ invariant. Our novel observation is that $\nu$ is defined over the fixed points, which can be identified as the effective boundary.
   In the next section, we will discuss about a  bulk-boundary correspondence on the level of K-theory.

\subsection{2d case}

For the 2d momentum spaces $X = \mathbb{S}^2$ or $\mathbb{T}^2$,  one has $\widetilde{KQ}(X)\cong \mathbb{Z}_2$,
which can be represented by  the Quaternionic Hilbert bundle $(\mathcal{H}, \Theta) \rightarrow (X, \tau)$.
Without loss of generality, we assume the Hilbert bundle $\mathcal{H}$ is of rank 2.

\begin{examp}
    If we define the 2d sphere in real coordinates,
$$
\mathbb{S}^2 = \{ (x, y, z)\in \mathbb{R}^3 \,\, | \,\, x^2 + y^2 + z^2 = 1 \}
$$
then the Hopf bundle over $\mathbb{S}^2$ can be represented by the projection $p \in M_2(C(\mathbb{S}^2))$ or the unitary $u = 2 p -1$,
$$
p = \frac{1}{2} \begin{pmatrix}
                   1 + z & x + iy \\
                   x- iy & 1-z
                \end{pmatrix}, \quad
u =   \begin{pmatrix}
                    z & x + iy \\
                   x- iy & -z
                \end{pmatrix}
$$
The first Chern character $Ch_1(p)$ gives the standard volume form on $\mathbb{S}^2$,
$$
Ch_1(p) = tr(pdpdp) = \frac{-i}{2} ( xdydz - ydxdz + zdxdy)
$$
The complex K-theory $K(\mathbb{S}^2) \cong \mathbb{Z}$ is basically generated by the Hopf bundle.

The  time reversal transformation on $\mathbb{S}^2$ is $\tau: (x, y, z) \mapsto (x, -y, -z)$,
and the first Chern character does not change under $\tau$, i.e., $Ch_1(p) = Ch_1(\tau^*(p))$,
where $\tau^*p$ is the pullback  of the Hopf bundle.
The transition function $w$ in this case can be defined by
$$
w : \mathbb{S}^2 \rightarrow SU(2); \quad (x, y, z) \mapsto \begin{pmatrix}
                                                              y+iz & x \\
                                                              -x & y-iz
                                                            \end{pmatrix}
$$
The fixed points of $\tau$ is $(\pm 1, 0,0)$, and $w(\pm 1, 0,0) = \begin{pmatrix}
                                                              0 & \pm 1 \\
                                                              \mp 1 & 0
                                                            \end{pmatrix}$.

\end{examp}

\begin{prop}
  The first Chern class of the Hilbert bundle $\pi: \mathcal{H} \rightarrow X$  is a 2 torsion,
  i.e.,
  \begin{equation}
   2  c_1(\mathcal{H}) = 0, \quad  c_1(\mathcal{H}) \in H^2(X, \mathbb{Z}),
  \end{equation}
\end{prop}
\begin{proof}
 Consider the pullback bundle $\tau^*\mathcal{H}$ over $X$,
 $$
   \xymatrix{
\tau^*\mathcal{H} \ar[d] \ar[r]^{\tau^*}  & \mathcal{H}  \ar[d] \\
X \ar[r]^{\tau} &  X }
   $$
 Since the time reversal transformation $\tau$ is orientation-reversing,
 the pullback bundle is isomorphic to  the conjugate bundle, i.e., $\tau^*\mathcal{H} \cong \overline{\mathcal{H}}$.
 On the other hand,  the time reversal operator $\Theta$ is an anti-unitary operator, so $\Theta$ induces a bundle isomorphism   $\mathcal{H} \cong  \overline{\mathcal{H}}$.
 In sum, 
 $$
  \mathcal{H}  \cong  \tau^*\mathcal{H} \cong \overline{\mathcal{H}}
 $$
 Therefore, the first Chern class of the Hilbert bundle $\mathcal{H}$ is a 2 torsion,
 $$
 c_1(\mathcal{H} ) = c_1( \overline{\mathcal{H}}  ) = - c_1(\mathcal{H} ), \quad i.e., \quad 2 c_1(\mathcal{H} ) = 0
 $$

\end{proof}

\begin{rmk}
If a Kramers pair is assumed to be a pair of global sections of the Hilbert bundle $\mathcal{H}$, then
$\mathcal{H}$ can be
decomposed into the sum of a trivializable  line bundle and the pullback of its conjugate,
$$
\mathcal{H} =   \mathcal{L} \oplus \tau^* \overline{\mathcal{L }}
$$
This is analogous to the decomposition of  a spin bundle.
In this case, the first Chern class of $\mathcal{H}$ is zero,
$$
c_1(\mathcal{H}) = c_1({\mathcal{L}}) + c_1(\tau^* \overline{\mathcal{L}}) =  c_1({\mathcal{L}}) +  c_1( {\mathcal{L}}) = 0
$$
since ${\mathcal{L}}$ has  trivial first Chern class. 
\end{rmk}

The Chern character $Ch: K(X) \otimes \mathbb{R} \cong H^{ev}_{dR}(X, \mathbb{R})$ establishes an isomorphism between K-theory and cohomology theory without torsion,
which maps the class of a vector bundle to the even part of de Rham cohomology with real coefficients.
The first Chern number $c_1$  is the integral
\begin{equation*}
  c_1  =  \frac{1}{2\pi}\int_X c_1(\mathcal{H}) = \frac{1}{2\pi}\int_X Ch_1(p) 
\end{equation*}
where $p$ is the projection representing the Hilbert bundle $\pi: \mathcal{H} \rightarrow X$.
It is well-known that the first Chern number  of topological insulators is zero. 
Hence the integral of the first Chern character modulo 2, i.e., $c_1$ mod $2$,  
\emph{cannot} give a local topological index formula for the   $\mathbb{Z}_2$ invariant in 2d.

We have to point out that it is an open problem to find a local formula for the topological $\mathbb{Z}_2$ invariant  in 2d.
 In \cite{Kell16}, the author gave a approach  by looping and lifting the problem from 2d to 3d.
Inspired by the Kane--Mele invariant, we propose a topological index formula below by reducing the dimension from 2d to 1d.

For simplicity, we consider the Brillouin zone (or torus) $\mathbb{T}^2$ as the momentum space, so the effective Brillouin zone $EBZ = \mathbb{S}^{1,1} \times I$  
and its boundary $\partial EBZ = \mathbb{S}^{1,1}_N \sqcup \mathbb{S}^{1,1}_S$.
Based on the weak $\mathbb{Z}_2$-structure of $(\mathbb{T}^2, \tau)$, $\mathbb{T}^2$ is decomposed into two copies of  $\mathbb{S}^{1,1} \times I$ and the Hilbert bundle $\pi: \mathcal{H} \rightarrow \mathbb{T}^2$
can be reconstructed from the restricted bundles $\mathcal{H}_{\mathbb{S}^{1,1} \times I}$ by the clutching map $w: \mathbb{S}^{1,1}_N \sqcup \mathbb{S}^{1,1}_S \rightarrow U(2)$.
We propose the topological index formula   on $\mathbb{T}^2$ as
\begin{equation}\label{2dTopInd}
   \frac{1}{4\pi^2}   \int_{\mathbb{S}^{1,1}_N \sqcup \mathbb{S}^{1,1}_S} tr(w^{-1}dw) =  [\frac{1}{2\pi}   \int_{\mathbb{S}^{1,1}_N} tr(w^{-1}dw) ] [\frac{1}{2\pi}   \int_{\mathbb{S}^{1,1}_S} tr(w^{-1}dw) ]
\end{equation}
 which is denoted by $ind_t(w)_{2 \rightarrow 1} $. 
Both integrals on the right hand side are $\mathbb{Z}_2$-valued as in the 3d case, so as a product $ind_t(w)_{2 \rightarrow 1} $  is also $\mathbb{Z}_2$-valued.

\begin{rmk}
As a gauge problem, the  information about the time reversal symmetry is encoded in the global gauge  transformation $w$.
   In odd dimensions, the topological $\mathbb{Z}_2$ invariant can be computed directly by the odd topological index  (i.e.,  WZW topological term), which is naturally $\mathbb{Z}_2$-valued due to the time reversal symmetry.
   In even dimensions, since the mod 2 reduction of the ordinary topological index cannot give a local formula for the topological $\mathbb{Z}_2$ invariant, we have to do a ``dimensional reduction'' and define
   it as the  WZW topological term over the boundary of the effective Brillouin zone ($\partial EBZ$), which is a variation of the odd topological index. 
   A similar  result for  the WZW topological term in the geometric setting of bundle gerbes can be found in \cite{G17}.

\end{rmk}

In this 2d case, we can also prove the mod 2 index theorem, i.e., the mod 2 analytical index equals the mod 2 topological index, if we collect the contributions from the 1d boundary of the effective Brillouin zone.
This is also closely related to our discussions about long exact sequences in K-theory, see  Example \ref{keyex} in  \S \ref{seqsec}.

\subsection{Index pairing} 

The  odd topological index is the integral of the odd Chern character $Ch(g)$ over the momentum space $X$, 
which can be viewed as a pairing between a K-theoretic class $[g] \in K^{-1}(X)$ and the fundamental class $[X]$,   
\begin{equation}\label{OddPair}
   ind_t(g) = \frac{1}{4\pi^2} \int_X Ch(g) = \langle [X], [g] \rangle 
\end{equation}
In a modern language, an index can be obtained by the index pairing between K-homology and K-theory.
Index pairing is very important for generalizations of index theory into noncommutative geometry. 
In this subsection, we reformulate the mod 2 index theorem as an index pairing
between KR-homology and KR-theory.

From a K-cycle $(\mathcal{A}, \mathscr{H} , D)$, one obtains the corresponding Fredholm module  $(\mathcal{A}, \mathscr{H} , F)$ by setting $F = D (1+D^2)^{-1/2}$.
By definition, the set of equivalence classes of Fredholm modules modulo unitary equivalence and homotopy equivalence defines the K-homology group, for details see \cite{HR00}.
In fact, the Dirac operator $D$ gives rise to a fundamental class $[D]$ in K-homology. 
From \S \ref{sec:KQcyc}, a Kramers pair is modeled by the KQ-cycle, and a localized Majorana zero mode is modeled by a  $KR_5$-cycle
$ (C^\infty(X), L^2(X, \mathcal{H}), D, \Theta)$,
which gives the fundamental class $[D]$ in the KR-homology $KR_5(X)$. 

On the other hand, the transition function $w: (X, \tau) \rightarrow (U(2), \sigma)$ completely determines the Quaternionic Hilbert bundle $\pi: (\mathcal{H}, \Theta) \rightarrow (X, \tau)$ 
(or the topological band theory in the presence of time reversal symmetry).
By the argument about spectral flow, the degree of $KQ$-theory is shifted by $-1$, and $w$ induces the K-theoretic class $[w]$ in $KQ^{-1} (X) = KR^{-5}(X)$. 

\begin{prop}
    The mod 2 index theorem of the topological $\mathbb{Z}_2$ invariant can be reformulated as an index pairing between KR-homology and KR-theory,
   \begin{equation}
      KR_5(X) \times KR^{-5}(X)  \rightarrow  KO^{-2}(pt); \quad    ([D], [w])\mapsto  \langle  [D] ,  [w] \rangle 
   \end{equation}
\begin{proof}
     From \eqref{OddPair}, the odd topological index $ ind_t(w)$ is a pairing,
   $$
   ind_t(w)= \langle  [X] ,  [w] \rangle 
   $$
   which can be used to compute the spectral flow $sf(d, w^{-1}dw)$ \cite{G93}. 
   After quantization, the trivial connection $d$ is replaced by the Dirac operator $D$, 
   $$
   sf(D, w^{-1}Dw) = ind_t(w)= \langle  [D] ,  [w] \rangle 
   $$
   which is the crucial part of the Atiyah--Singer index theorem connecting analysis and geometry. 
   The analytical index   $ind_a(\tilde{H})$ can be computed by the spectral flow  $ sf(D, w^{-1}Dw)$ modulo 2, 
   so the index pairing between $[D]$ and $[w]$ is a fancy way to express the mod 2 index theorem
    $$
   ind_a(\tilde{H}) = ind_t(w)
   $$ 
   Taking time reversal symmetry into account, $[D]$ is the class of the $KR_5$-cycle, so it represents the fundamental class in $ KR_5(X)$.
   On the other hand,  $[w]$ falls into $KR^{-5}(X)$ and represents the class of the Quaternionic Hilbert bundle.  
   In sum, the index pairing $ \langle  [D] ,  [w] \rangle $ gives another way to compute the topological $\mathbb{Z}_2$ invariant. 
\end{proof}

\end{prop}

\section{Bulk-boundary correspondence}\label{sec:Bbcorr}

In this section, we will focus on the bulk-boundary correspondence,
which is a generalized index map in  KK-theory.
The bulk is the momentum space $X$ and the boundary will be identified as the fixed points of the time reversal symmetry.
So at the level of K-theory,  the bulk theory is modeled by the $KQ$-theory of $X$,
and the boundary theory is described by the  $KO$-theory of $X^\tau$.
Therefore the bulk-boundary correspondence is a map connecting $KR$ (or $KO$)-groups, which can be realized by a $KKR$ (or $KKO$)-cycle.

The bulk-boundary correspondence of topological insulators is an active research area in mathematical physics, one can find many recent works in this field,
for example \cite{ BKR16, GP13, K16, MT16} and many others. Here we take a slightly different approach and propose that the effective boundary of a topological insulator
is described by the fixed points in the momentum space.

\subsection{Effective boundary}
In physics, the bulk--boundary correspondence establishes the equivalence between the effective (quantum) field theories of the bulk and boundary. 
As different aspects of the topological $\mathbb{Z}_2$ invariant, the topological $\mathbb{Z}_2$ index  and the  Kane--Mele invariant are equivalent,
which is the geometric form of the equivalence between two field theories. 
The topological index  is an integral of the (odd) Chern character over the bulk, and the Kane--Mele invariant has a Pfaffian formalism derived from the effective fermionic field theory. 
In other words, the bulk theory is an action functional of a specific gauge field  (i.e., the WZW topological term), and the Kane--Mele invariant is an effective quantum field theory defined over the set of fixed points. 
Inspired by the Kane--Mele invariant,
we propose that the effective boundary is defined as the fixed points, which is different from the geometric boundary of the bulk. 
Hence the equivalence between the topological index and the Kane--Mele invariant is 
the genuine bulk-boundary correspondence on the level of $\mathbb{Z}_2$ invariant.

The phase portrait
of a topological insulator is essentially determined by its behavior on the boundary.
In mathematics, let us look at the mod 2 analytical index  again,
and identify the fixed points as the effective boundary. 
In the Hilbert bundle $\pi: \mathcal{H} \rightarrow X$, the fixed points  divide the momentum space $X$
into different coordinate patches. The chiral states $\phi$ and $\Theta \phi$ in a Kramers pair
could intersect with each other and change the chirality only at a fixed point.
For a localized Majorana zero mode $(\phi_0, \Theta \phi_0)$ around a fixed point $x \in X^\tau$,
the sign of the eigenvalue of $\phi_0$ or $\Theta \phi_0$ can only change when passing through the fixed point $x$,
and the spectral flow of a chiral zero mode, i.e., $\psi_0$ or $\Theta \psi_0$, counts  the existence of that Majorana zero mode.
Finally, the topological $\mathbb{Z}_2$ invariant is the parity of Majorana zero modes, which can be computed by
the mod 2 spectral flow. The argument about spectral flow suggests that the fixed points can be identified as the effective boundary, 
since charge transport via spectral flow only happens at the boundary.

At each fixed point $x \in X^\tau$,  the (mod 2) spectral flow of a chiral zero mode around $x$ gives rise to a local analytical index,
which takes the value in
$KO^{-2}(x) = \mathbb{Z}_2$.
In total, the topological $\mathbb{Z}_2$ invariant belongs to 
\begin{equation}
    \oplus_{x \in X^\tau}KO^{-2}(x)= KO^{-2}(X^\tau)
   \stackrel{\sum}{\rightarrow} KO^{-2}(pt) 
\end{equation}
where the summation map $\Sigma$ is basically the push-forward induced by the projection to an abstract point $p:X^{\tau}\to pt$. 
Physically, this summation represents the collective effect -- i.e., summing up  mod 2 -- in the end.
 It can also be viewed as performing a discrete form of integration, and then counting the parity.
This explains our identification of the fixed points $X^\tau$ as the boundary
through which information from the bulk can be exchanged. The effective boundary  $X^\tau$   gives rise to a family
index theorem, to get an index number, it is better to further project it to an abstract point as in the construction of the topological index map by Atiyah.

Furthermore, the effective boundary will be reduced to  the fixed point with ``top codimension'' inspired by the topological index map.
By the Kane--Mele invariant, the topological $\mathbb{Z}_2$ invariant is a  collective effect over all fixed points.
It is convenient to realize the above abstract point  as a specific fixed point $x_0$
 so that the topological $\mathbb{Z}_2$ invariant takes the value in $ KO^{-2}(x_0)$.
In the decomposition of $\mathbb{S}^{1,d}$ in \S\ref{quatsec} or  $\mathbb{T}^d$ in the next example,
this specific fixed point can be taken as the fixed point corresponding to the first summand, that is the point of ``top codimension''.
This fixed point $x_0 \in X^\tau$ with top codimension is naturally identified with the sphere with top dimension by the   KR-theory of spheres in Eq.\eqref{KRSd}.
For a $d$-dimensional momentum space $X$, the topological index map in \S \ref{sec: Tindmap} can be viewed as a map from the KQ-theory of $X$ to the KR-theory of $\mathbb{S}^{1,d}$,
\begin{equation}
ind_t: {KQ}^{2-d}(X) =  {KR}^{-d-2}(X) \mapsto KO^{-2}(x_0) \cong   \widetilde{KR}^{-d-2}(\mathbb{S}^{1,d})
\end{equation}

\begin{examp}
    The $KR$-theory   of $\mathbb{T}^d$ can be computed  based on the decomposition of the fixed points
$$
(\mathbb{T}^d)^\tau = \oplus _{k =0}^d \binom{d}{k} \{ pt \}
$$
where $\binom{d}{k}$ is the binomial coefficient,  so that
$$
KR^{-j}(\mathbb{T}^d) = \oplus_{n =0}^d \binom{d}{k} KO^{-j+k}(pt)
$$
The spheres ${S}^k$ correspond to the fixed points with different codimensions
so that the torus has the stable homotopy splitting \cite{FM13}
$$
\mathbb{T}^d \sim_s \vee_{k =0}^d \binom{d}{k} {S}^k
$$
As the strong topological $\mathbb{Z}_2$ invariant, it only depends on the fixed point $x_0$ with top codimension and falls into
\begin{equation*}
   KO^{-j+d}(x_0)  =  \widetilde{KR}^{-j} (\mathbb{S}^{1,d})
\end{equation*}
where $j$ is fixed by the condition $-j+d = -2$.
\end{examp}

\subsection{Correspondence}

Once the effective boundary is identified as the set of fixed points (or the fixed point with top codimension), we can talk about
the bulk-boundary correspondence on the level of K-theory.
The bulk-boundary correspondence  is supposed to be a KK-cycle (as a correspondence)  connecting
the bulk and boundary theory if both  theories are modeled by K-theory (or K-homology).

Let $X$ be the momentum space and $x_0 \in X^\tau$ be the fixed point with top codimension.
The bulk theory is described by the $KQ$-theory of $X$,
the boundary theory is modeled by  $ KO^{-2}(X^\tau)$ (or $ KO^{-2}(x_0)$) from the above discussions. 
When $X$ is the torus $\mathbb{T}^d$, one has the Baum--Connes isomorphism for $\Gamma = \mathbb{Z}^d$ in \S \ref{BCIso}. 
Let us  look into this example again, and then we will pose a generalization for the general case.

\begin{prop}
   The bulk-boundary correspondence on the level of K-theory for the torus $\mathbb{T}^d$ is a surjective map from $KQ^{2-d}(\mathbb{T}^d) = KR^{-d-2}(\mathbb{T}^d)$ to $KO^{-2}(x_0)$, 
   which can be realized by a composition map, 
   \begin{equation}
        i_0^* \circ PD \circ \alpha: {KR}^{-d-2} (\mathbb{T}^d) \cong {KO}_{d+2}(\mathbb{T}^d) \cong {KO}^{-2}(\mathbb{T}^d) \rightarrow KO^{-2} (x_0)
   \end{equation}
   where $\alpha$ is the Dirac isomorphism, PD is the   Poincar\'{e} duality between $KO$-homology and $KO$-theory, and $i_0^*$ is the localization. 
\end{prop} 

\begin{proof}
  
If $i_0: x_0 \hookrightarrow \mathbb{T}^d$ is the inclusion map, then it induces the restriction  map (or localization) in $KO$-theory,
$$i_0^*:  KO^{-j}(\mathbb{T}^d) \rightarrow KO^{-j} (x_0)$$ 
The Baum--Connes isomorphism or the dual Dirac isomorphism
\begin{equation*}
       KO_i(\mathbb{T}^d) = KO_i^{\mathbb{Z}^d} (\mathbb{R}^d) \simeq KO_i(C^*(\mathbb{Z}^d, \mathbb{R})) = KO_i(C(\mathbb{T}^d, \tau)) = KR^{-i} (\mathbb{T}^d, \tau) 
\end{equation*}
can be realized by an invertible real KK-class $\beta$ in $KKO(C_\mathbb{R}(\mathbb{T}^d), C(\mathbb{T}^d, \tau))$
connecting the  KO-homology of $C_\mathbb{R}(\mathbb{T}^d) $ and the KO-theory of $ C(\mathbb{T}^d, \tau)$, i.e.,
$$
KO_i(\mathbb{T}^d) \times KKO(C_\mathbb{R}(\mathbb{T}^d), C(\mathbb{T}^d, \tau)) \xrightarrow{\simeq} KO_i(C(\mathbb{T}^d, \tau))
$$ 
Let $\alpha$ be the inverse  of $\beta$,
that is, $\alpha \in KKO( C(\mathbb{T}^d, \tau), C_\mathbb{R}(\mathbb{T}^d)) $ such that $\alpha \circ \beta = id$,
so $\alpha$ realizes the Dirac isomorphism
$$
KR^{-i} (\mathbb{T}^d, \tau) \times KKO( C(\mathbb{T}^d, \tau), C_\mathbb{R}(\mathbb{T}^d)) \xrightarrow{\simeq} KO_i(\mathbb{T}^d) 
$$ 
Hence the map from the bulk K-theory to the boundary K-theory, denoted by $\eta :   KR^{-d-2}(\mathbb{T}^d)  \rightarrow  KO^{-2}(x_0)$, can be realized by $\eta = i_0^* \circ PD \circ \alpha$.
\end{proof}

   \begin{examp}
   When $d= 2$, the bulk-boundary correspondence of a 2d topological insulator is the isomorphism
   \begin{equation*}
      \eta:  \widetilde{KQ}(\mathbb{T}^2) = \widetilde{KR}^{-4} (\mathbb{T}^2)   \xrightarrow{\simeq} KO^{-2} (x_0)
   \end{equation*}
   \end{examp}

   \begin{examp}
   When $d= 3$, the bulk-boundary correspondence of a 3d topological insulator is the surjective map
   \begin{equation*}
      \eta:  {KQ}^{-1}(\mathbb{T}^3) =  {KR}^{-5} (\mathbb{T}^3)   \xrightarrow{} KO^{-2} (x_0)
   \end{equation*}
   \end{examp}

 The above bulk-boundary correspondence  $\eta: KQ^{2-d}(\mathbb{T}^d) \rightarrow KO^{-2}(x_0)$ is
 another way to realize the topological index map, where the KK-cycle $\alpha$ plays a key role.
 As a generalization, the bulk-boundary correspondence in $KR$-theory is a $KKR$-cycle in $KKR^{d}(X, X^\tau)$ for $d = \dim(X) $ realizing the topological index map in KK-theory,
 \begin{equation}
    KR^{-d-2}(X, \tau) \times KKR^{d}(X, X^\tau) \rightarrow KO^{-2}(X^\tau)   \stackrel{\sum}{\rightarrow} KO^{-2}(x_0)
 \end{equation}

On the other hand, the bulk theory can be modeled by $K$-cycles or $K$-homology,
such an idea has been carried out for example in \cite{BCR16}. So the bulk-boundary
correspondence in $KR$-homology  has the  form,
\begin{equation}
   KKR(X^\tau, X) \times KR_2(X) \rightarrow KO_{2}(X^\tau) \xrightarrow{\simeq}   KO^{-2}(X^\tau)  \stackrel{\sum}{\rightarrow} KO^{-2}(x_0)
\end{equation}
where the isomorphism is the Poincar\'{e} duality. This bulk-boundary correspondence on the level of K-homology is viewed as a generalized
analytical index map.

In sum, the bulk-boundary correspondence provides a new perspective to view the
mod 2 index theorem in KK-theory,
 \begin{equation}
   \xymatrixcolsep{5pc}\xymatrix{
  KKR(X^\tau, X) \times KR_2(X)   \ar[r]^-{ind_a}   & KO^{-2}(x_0)  \ar[d]^{=}    \\
      KQ^{2-d}(X) \times KKR^{d}(X, X^\tau) \ar[r]^-{ind_t}  &   KO^{-2}(x_0)   }
\end{equation}

\section*{Acknowledgements} The authors would like to thank Jonathan Rosenberg for exciting discussions on KR-theory.
RK  thankfully  acknowledges  support  from  the  Simons foundation under
collaboration grant \# 317149 and BK thankfully  acknowledges  support  from  the NSF
under the grants PHY-0969689 and PHY-1255409.

\nocite{*}
\bibliographystyle{plain}
\bibliography{IndKQ}

\begin{thebibliography}{10}

\bibitem{AZ97}
A.~Altland and M.~Zirnbauer.
\newblock Nonstandard symmetry classes in mesoscopic normal-superconducting
  hybrid structures.
\newblock {\em Phys. Rev. B}, 55:1142, 1997.

\bibitem{A66}
M.~Atiyah.
\newblock K-theory and reality.
\newblock {\em Quart. J. Math.}, 17(2):367--386, 1966.

\bibitem{A76}
M.~Atiyah.
\newblock Vector bundles on projective 3-space.
\newblock {\em Inventiones Math.}, 35:131--153, 1976.

\bibitem{AS69}
M.~Atiyah and I.M. Singer.
\newblock Index theory for skew-adjoint {F}redholm operators.
\newblock {\em Inst. Hautes Etudes Sci. Publ. Math.}, 37:5--26, 1969.

\bibitem{AS71}
M.~Atiyah and I.M. Singer.
\newblock The index of elliptic operators: V.
\newblock {\em Ann. of Math.}, 93:139--149, 1971.

\bibitem{ASV13}
J.C. Avila, H.~Schulz-Baldes, and C.~Villegas-Blas.
\newblock Topological invariants of edge states for periodic two-dimensional
  models.
\newblock {\em Math. Phys., Anal. Geom.}, 16:136--170, 2013.

\bibitem{ADW91}
S.~Axelrod, S.~Della Pietra, and E.~Witten.
\newblock Geometric quantization of {C}hern--{S}imons gauge theory.
\newblock {\em J. Diff. Geo.}, 33:787--902, 1991.

\bibitem{BD82}
P.~Baum and R.G. Douglas.
\newblock K-homology and index theory.
\newblock In {\em Operator Algebras and Applications}, volume~38 of {\em Proc.
  Sympos. Pure and Appl. Math.}, pages 117--173. AMS. Providence RI, 1982.

\bibitem{BHS07}
P.~Baum, N.~Higson, and T.~Schick.
\newblock On the equivalence of geometric and analytic {K}-homology.
\newblock {\em Pure and Applied Math. Quaterly}, 3:1--24, 2007.

\bibitem{BK04}
P.~Baum and M.~Karoubi.
\newblock On the {B}aum--{C}onnes conjecture in the real case.
\newblock {\em Quart. J. Math.}, 55:231--235, 2004.

\bibitem{BCR16}
C.~Bourne, A.L. Carey, and A.~Rennie.
\newblock A noncommutative framework for topological insulators.
\newblock {\em Rev. Math. Phys.}, 28:1650004, 2016.

\bibitem{BKR16}
C.~Bourne, J.~Kellendonk, and A.~Rennie.
\newblock The {K}-theoretic bulk-edge correspondence for topological
  insulators.
\newblock 2016.
\newblock arXiv: 1604.02337.

\bibitem{CPS16}
A.~Carey, J.~Phillips, and H.~Schulz-Baldes.
\newblock Spectral flow for skew-adjoint {F}redholm operators.
\newblock 2016.
\newblock arXiv: 1604.06994.

\bibitem{CGLW13}
X.~Chen, Z.C. Gu, Z.X. Liu, and X.G. Wen.
\newblock Symmetry protected topological orders and the group cohomology of
  their symmetry group.
\newblock {\em Phys. Rev. B}, 87:155114, 2013.

\bibitem{C95}
A.~Connes.
\newblock Noncommutative geometry and reality.
\newblock {\em J. Math. Phys.}, 36:11, 1995.

\bibitem{DG15}
G.~{De Nittis} and K.~Gomi.
\newblock Classification of ``{Q}uaternionic'' {B}loch-bundles: Topological
  insulators of type {AII}.
\newblock {\em Commun. Math. Phys.}, 339:1--55, 2015.

\bibitem{DS16}
G.~{De Nittis} and H.~Schulz-Baldes.
\newblock Spectral flows associated to flux tubes.
\newblock {\em Annales Henri Poincare}, 17:1--35, 2016.

\bibitem{FM13}
D.~Freed and G.M. Moore.
\newblock Twisted equivariant matter.
\newblock {\em Annales Henri Poincar{\'e}}, 14(8):1927--2023, 2013.

\bibitem{FK06}
L.~Fu and C.~Kane.
\newblock Time reversal polarization and a $\mathbb{Z}_2$ adiabatic spin pump.
\newblock {\em Phys. Rev. B}, 74:195312, 2006.

\bibitem{FKM07}
L.~Fu, C.~Kane, and E.~Mele.
\newblock Topological insulators in three dimensions.
\newblock {\em Phys. Rev. Lett.}, 98:106803, 2007.

\bibitem{FSFF12}
T.~Fukui, K.~Shiozaki, T.~Fujiwara, and S.~Fujimoto.
\newblock Bulk-edge correspondence for {C}hern topological phases: A viewpoint
  from a generalized index theorem.
\newblock {\em J. Phys. Soc. Jpn.}, 81:114602, 2012.

\bibitem{G17}
K.~Gawedzki.
\newblock 2d {F}u--{K}ane--{M}ele invariant as {W}ess--{Z}umino action of the
  sewing matrix.
\newblock {\em Lett. Math. Phys.}, 107(4):733--755, 2017.

\bibitem{G93}
E.~Getzler.
\newblock The odd {C}hern character in cyclic homology and spectral flow.
\newblock {\em Topology}, 32(3):489--507, 1993.

\bibitem{GP13}
G.M. Graf and M.~Porta.
\newblock Bulk-edge correspondence for two-dimensional topological insulators.
\newblock {\em Commun. Math. Phys.}, 324(3):851--895, 2013.

\bibitem{GS15}
J.~Grossmann and H.~Schulz-Baldes.
\newblock Index pairings in presence of symmetries with applications to
  topological insulators.
\newblock {\em Commum. Math. Phys.}, 343(2):477--513, 2015.

\bibitem{HMT16}
K.~Hannabuss, V.~Mathai, and G.C. Thiang.
\newblock T-duality simplifies bulk-boundary correspondence: the general case.
\newblock 2016.
\newblock arXiv: 1603.00116.

\bibitem{HK10}
M.Z. Hasan and C.L. Kane.
\newblock Colloquium: Topological insulators.
\newblock {\em Reviews of Modern Physics}, 82(4):3045, 2010.

\bibitem{HR00}
N.~Higson and J.~Roe.
\newblock {\em Analytic {K}-homology}.
\newblock Oxford Science Publications, 2000.

\bibitem{KM05}
C.~Kane and E.~Mele.
\newblock $\mathbb{Z}_2$ topological order and the quantum spin {H}all effect.
\newblock {\em Phys. Rev. Lett.}, 95:146802, 2005.

\bibitem{KM0501}
C.~Kane and E.~Mele.
\newblock Quantum spin {H}all effect in graphene.
\newblock {\em Phys. Rev. Lett.}, 95:226801, 2005.

\bibitem{Karoubi}
Max Karoubi.
\newblock {\em {$K$}-theory}.
\newblock Classics in Mathematics. Springer-Verlag, Berlin, 2008.
\newblock An introduction, Reprint of the 1978 edition, With a new postface by
  the author and a list of errata.

\bibitem{K80}
G.~Kasparov.
\newblock The operator {K}-functor and extensions of ${C}^\ast$-algebras.
\newblock {\em Izv. Akad. Nauk. SSSR Ser. Mat.}, 44:571--636, 1980.

\bibitem{KK16}
H.~Katsura and T.~Koma.
\newblock The $\mathbb{Z}_2$ index of disordered topological insulators with
  time reversal symmetry.
\newblock {\em J. Math. Phys.}, 57:021903, 2016.

\bibitem{KLW1601}
R.~Kaufmann, D.~Li, and B.~Wehefritz-Kaufmann.
\newblock Noncommutative topological $\mathbb{Z}_2$ invariant.
\newblock 2016.
\newblock arXiv: 1605.09470.

\bibitem{KLW15}
R.~Kaufmann, D.~Li, and B.~Wehefritz-Kaufmann.
\newblock Notes on topological insulators.
\newblock {\em Rev. Math. Phys.}, 28(10):1630003, 2016.

\bibitem{KLW16}
R.~Kaufmann, D.~Li, and B.~Wehefritz-Kaufmann.
\newblock The {S}tiefel--{W}hitney theory of topological insulators.
\newblock 2016.
\newblock arXiv: 1604.02792.

\bibitem{K15}
J.~Kellendonk.
\newblock On the ${C}^*$-algebraic approach to topological phases for
  insulators.
\newblock 2015.
\newblock arXiv: 1509.06271.

\bibitem{Kell16}
J.~Kellendonk.
\newblock Cyclic cohomology for graded ${C}^{\ast,r}$-algebras and its pairings
  with van daele {K}-theory.
\newblock 2016.
\newblock arXiv: 1607.08465.

\bibitem{K09}
A.~Kitaev.
\newblock Periodic table for topological insulators and superconductors.
\newblock {\em AIP Conf. Proc.}, 1134:22--30, 2009.

\bibitem{K16}
Y.~Kubota.
\newblock Controlled topological phases and bulk-edge correspondence.
\newblock {\em Commun. Math. Phys.}, pages 1--33, 2016.

\bibitem{LM90}
H.~B. Lawson and M.~Michelsohn.
\newblock {\em Spin geometry}.
\newblock Princeton University Press, 1990.

\bibitem{L88}
J.~Lott.
\newblock Real anomalies.
\newblock {\em J. Math. Phys.}, 29:1455--1464, 1988.

\bibitem{MT16}
V.~Mathai and G.C. Thiang.
\newblock T-duality simplifies bulk-boundary correspondence.
\newblock {\em Commun. Math. Phys.}, 345(2):675--701, 2016.

\bibitem{MB07}
J.~Moore and L.~Balents.
\newblock Topological invariants of time-reversal-invariant band structures.
\newblock {\em Phys. Rev. B.}, 75:121306, 2007.

\bibitem{P96}
J.~Phillips.
\newblock Self-adjoint {F}redholm operators and spectral flow.
\newblock {\em Canad. Math. Bull.}, 39:460--467, 1996.

\bibitem{PLB13}
E.~Prodan, B.~Leung, and J.~Bellissard.
\newblock The non-commutative nth-{C}hern number $(n \geq 1)$.
\newblock {\em J. Phys. A: Math. Theor.}, 46:485202, 2013.

\bibitem{PS16}
E.~Prodan and H.~Schulz-Baldes.
\newblock {\em Bulk and boundary invariants for complex topological insulators:
  From {K}-theory to physics}.
\newblock Springer, Berlin, 2016.

\bibitem{QHZ08}
X.L. Qi, T.~Hughes, and S.C. Zhang.
\newblock Topological field theory of time-reversal invariant insulators.
\newblock {\em Phys. Rev. B}, 78:195424, 2008.

\bibitem{QZ11}
X.L. Qi and S.C. Zhang.
\newblock Topological insulators and superconductors.
\newblock {\em Rev. Mod. Phys.}, 83:1057--1111, 2011.

\bibitem{R14}
J.~Rosenberg.
\newblock Real {B}aum--{C}onnes assembly and {T}-duality for torus
  orientifolds.
\newblock {\em J. Geom. Phys.}, 89:24--31, 2014.

\bibitem{DFN15}
S.~Das Sarma, M.~Freedman, and C.~Nayak.
\newblock Majorana zero modes and topological quantum computation.
\newblock {\em Quantum Information}, 1:15001, 2015.

\bibitem{SRFL09}
A.~Schnyder, S.~Ryu, A.~Furusaki, and A.~Ludwig.
\newblock Classification of topological insulators and superconductors.
\newblock {\em AIP Conf. Proc.}, 1134:10--21, 2009.

\bibitem{S16}
H.~Schulz-Baldes.
\newblock Topological insulators from the perspective of non-commutative
  geometry and index theory.
\newblock 2016.
\newblock arXiv: 1607.04013.

\bibitem{S13}
S.Q. Shen.
\newblock {\em Topological Insulators: {D}irac Equation in Condensed Matters}.
\newblock Springer, 2013.

\bibitem{T16}
G.C. Thiang.
\newblock On the {K}-theoretic classification of topological phases of matter.
\newblock {\em Annales Henri Poincare}, 17(4):757--794, 2016.

\bibitem{VFG01}
J.~C. Varilly, H.~Figueroa, and J.~M. Gracia-Bondia.
\newblock {\em Elements of Noncommutative Geometry}.
\newblock Birkh{\"a}user, 2001.

\bibitem{WQZ10}
Z.~Wang, X.L. Qi, and S.C. Zhang.
\newblock Equivalent topological invariants of topological insulators.
\newblock {\em New J. Phys.}, 12:065007, 2010.

\bibitem{Wigner}
E.~Wigner.
\newblock \"{U}ber die {O}peration der {Z}eitumkehr in der {Q}uantenmechanik.
\newblock {\em Nachrichten von der Gesellschaft der Wissenschaftern zu
  G\"{o}ttingen}, pages 546--559, 1932.

\bibitem{W09}
F.~Wilczek.
\newblock Majorana returns.
\newblock {\em Nature Physics}, 5:614--618, 2009.

\bibitem{W82}
E.~Witten.
\newblock An {SU}(2) anomaly.
\newblock {\em Phys. Lett. B}, 117(5):324--328, 1982.

\bibitem{W84}
E.~Witten.
\newblock Non-abelian bosonization in two dimensions.
\newblock {\em Commun. Math. Phys.}, 92(4):455--472, 1984.

\bibitem{W16}
E.~Witten.
\newblock Fermion path integrals and topological phases.
\newblock {\em Rev. Mod. Phys.}, 88:035001, 2016.

\end{thebibliography}

\end{document}